\documentclass[11pt]{article}

\usepackage{amssymb} 
\usepackage{latexsym} 
\usepackage{amsmath}
\usepackage{appendix}
\usepackage{color}
\usepackage{epsfig}
\usepackage{float}
\usepackage{graphicx}
\usepackage[all]{xy}
\input xy
\usepackage[mathcal]{eucal}
\usepackage{enumerate}
\usepackage{anysize}
\usepackage[english]{babel}
\usepackage{hyperref}

\newtheorem{definition}{Definition}

\newtheorem{remark}{Remark}

\newtheorem{theorem}{Theorem}

\newtheorem{lemma}{Lemma}
\newtheorem{claim}{Claim}

\newtheorem{coro}{Corollary}

\newtheorem{proofTh} {Proof.}

\newcommand{\CVD} {\hspace*{\fill}$\Box$}

\newenvironment{proof}{\begin{proofTh}\em}{\CVD\end{proofTh}}

\def\cadre{$$\vcenter\bgroup\advance\hsize by -2em\noindent
             \refstepcounter{equation}(\theequation)~\ignorespaces}
\makeatletter
\def\endcadre{\egroup\eqno$$\global\@ignoretrue\par\vspace{-2mm}}
\def\ncadre{$$\vcenter\bgroup\advance\hsize by -2em\noindent
             \ignorespaces}
\makeatletter
\def\endncadre{\egroup\eqno$$\global\@ignoretrue}

\newcommand{\comment}[1]{}
\newcommand{\mybreak} {\par\vspace{2mm}\noindent}

\newcommand{\mv}[1] {\mathsf{#1}}

\makeatletter
\def\imod#1{\allowbreak\mkern10mu({\operator@font mod}\,\,#1)}
\makeatother

\newcommand{\A} {\mathsf{A}}
\newcommand{\D} {\mathsf{D}}

\newcounter{progcount}
\newcounter{linecount}[progcount]

{
    \end{tabbing}
    \end{minipage}\\[0.5ex]
}

\ignorespaces

\newcommand{\dist}        {\textrm{dist}}
\title{A tight relation between series--parallel graphs and Bipartite Distance Hereditary graphs.}

\author{Nicola Apollonio\footnote{Istituto per le Applicazioni del
Calcolo,~M.~Picone,~v.~dei Taurini 19, 00185 Roma, Italy.
\texttt{nicola.apollonio@cnr.it}.\ Supported by the Italian National Research Council (C.N.R.) under national
research project ``MATHTECH''.} \and {Massimiliano
Caramia\footnote{Dipartimento di Ingegneria dell'Impresa,
Universit\`a di Roma ``Tor Vergata'', v.\ del Politecnico 1, 00133
Roma, Italy. \texttt{caramia@disp.uniroma2.it}}}
\and
{Paolo Giulio Franciosa\footnote{Dipartimento di Scienze Statistiche,
Sapienza Universit\`a di Roma, p.le Aldo Moro 5, 00185
Roma, Italy. \texttt{paolo.franciosa@uniroma1.it}}}
\and {Jean-Fran\c{c}ois Mascari\footnote{Istituto per le Applicazioni del
Calcolo, M. Picone, v. dei Taurini 19, 00185 Roma, Italy.
\texttt{g.mascari@iac.cnr.it}.\ Supported by the Italian National Research Council (C.N.R.) under national
research project ``MATHTECH''.}}}

\date{}

\begin{document}

\maketitle

\begin{abstract}
Bandelt and Mulder's structural characterization of Bipartite
Distance Hereditary graphs asserts that such graphs can be built
inductively starting from a single vertex and by repeatedly adding
either pending vertices or twins (i.e., vertices with the same
neighborhood as an existing one). Dirac and Duffin's structural characterization
of 2--connected series--parallel graphs asserts that such graphs
can be built inductively starting from a single edge by adding
either edges in series or in parallel. In this paper we prove that the two constructions are the same construction when bipartite graphs are viewed as the fundamental graphs of a graphic matroid. We then apply the result to re-prove known results concerning bipartite distance hereditary graphs and series--parallel graphs, to characterize self--dual outer-planar graphs and, finally, to provide a new class of polynomially-solvable instances for the integer multi commodity flow of maximum value.
\mybreak
\end{abstract}

\section{Introduction}\label{sec:introduction}
\emph{Distance Hereditary} graphs are graphs with the
\emph{isometric property}, i.e., the distance function of a
distance hereditary graph is inherited by its connected induced
subgraphs. This important class of graphs was introduced and
thoroughly investigated by Howorka in \cite{how1,how2}. A Bipartite Distance
Hereditary (BDH for shortness) graph is a distance hereditary graph which is bipartite.
Such graphs can be constructed starting from a single vertex by means of the following two operations~\cite{bm}:
\begin{enumerate}[(BDH1)]
\item\label{com:pendant} adding a
\emph{pending vertex}, namely a vertex adjacent exactly to an existing vertex; 
\item\label{com:twin} adding a \emph{twin} of an existing vertex, namely adding a vertex and making it adjacent to all the neighbors of an existing vertex.  
\end{enumerate}
Taken together the two operations above will be referred to as Bandelt and Mulder's construction.
\mybreak
A graph is \emph{series--parallel} \cite{bogikano}, if it does not contain the complete graph $K_4$ as a minor; equivalently, if it does not contain a subdivision of $K_4$. This is Dirac's~\cite{dirac} and Duffin's~\cite{duffin} characterization by forbidden minors. Since both $K_5$ and $K_{3,3}$ contain a subdivision of $K_4$, by Kuratowski's Theorem any series--parallel graph is planar. Like BDH graphs, series--parallel graphs admit a constructive characterization which justifies their name: a connected graph is series--parallel if it can be constructed starting from a single edge by means of the following two operations:
\begin{enumerate}[{\rm (SP1)}]
\item\label{com:parallel} adding an edge with the same end-vertices as an existing
one (\emph{parallel extension})
\item\label{com:series} subdividing an existing edge by the insertion of a new
vertex (\emph{series
extension}.)
\end{enumerate}
Taken together the two operations above will be referred to as Duffin's construction---here and throughout the rest of the paper we consider only 2--connected series--parallel graphs which can be therefore obtained starting from a pair of a parallel edges rather than a single edge--.
\begin{figure}[H]
	\begin{center}
		\includegraphics[width=10cm]{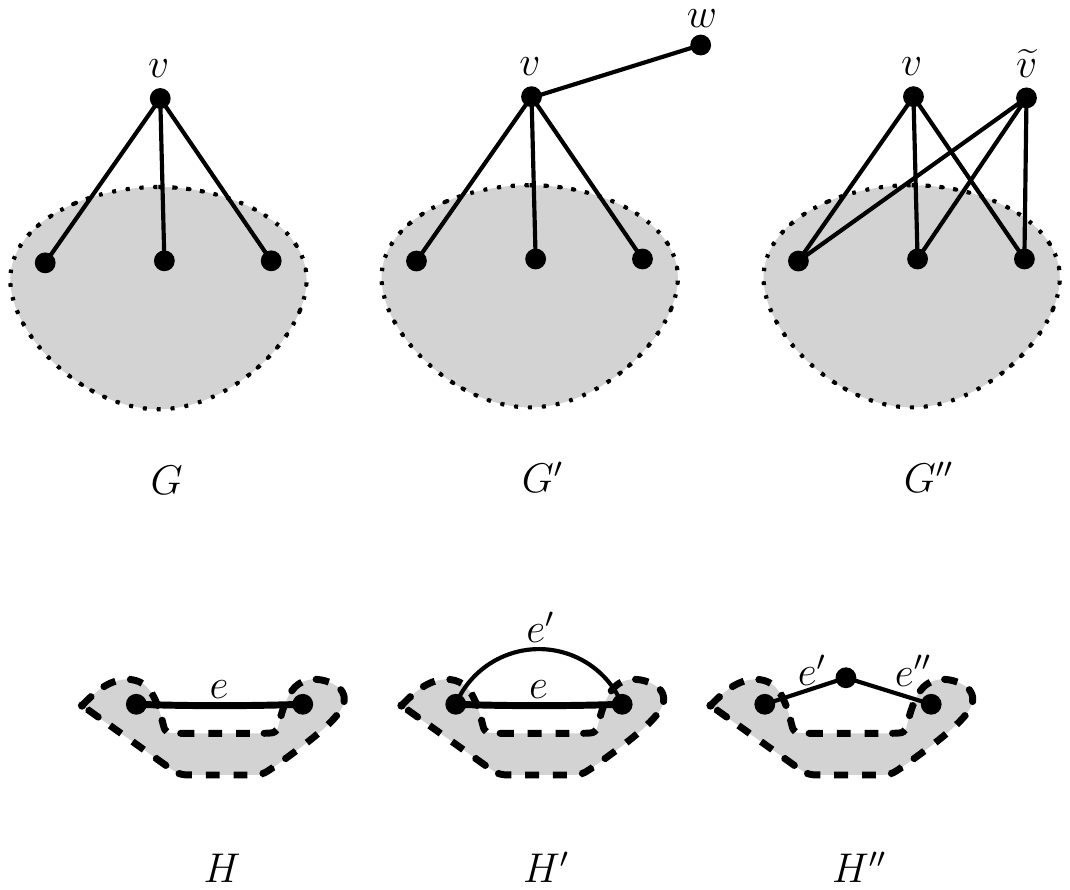}
	\end{center}
	\caption{Bandelt and Mulder's and Duffin's constructions: graph $G'$ and $G''$ are the results of adding a pending vertex $w$ at $v$ or a twin $\widetilde{v}$ to vertex $v$, respectively. Graph $H'$ and $H''$ are the results of adding an edge $e'$ parallel to $e$ or subdividing edge $e$, respectively.}
	\protect\label{fi:bamudu}
	\noindent\hrulefill%
\end{figure}

The close resemblance between operations (BDH\ref{com:pendant})$\div$(BDH\ref{com:twin}) and operations (SP\ref{com:parallel})$\div$(SP\ref{com:series}) is apparent (see Figure~\ref{fi:bamudu} for an illustration of both the construction).\ It becomes even more apparent after our Theorem~\ref{thm:main}, which establishes that the constructions defining BDH and series--parallel graphs, namely, Bandelt and Mulder's construction and Duffin's construction, are essentially the same construction in a sense made precise after Section~\ref{sec:prel}. 
\mybreak
The intimate relationship between BDH graphs and series--parallel graphs was already observed by Ellis-Monhagan and Sarmiento in \cite{mosa}. The authors, motivated by the aim of finding polynomially
computable classes of instances for the
\emph{vertex--nullity interlace polynomial} introduced by Arratia,
Bollob\'{a}s and Sorkin in \cite{abs}, under the name of
\emph{interlace polynomial}, related the two classes of graphs via a topological construction involving the so called \emph{medial graph} of a planar graph. By further relying on the relationships between the \emph{Martin polynomial} and the \emph{symmetric Tutte polynomial} of a planar graph, they proved a relation between the the \emph{symmetric Tutte polynomial} of a planar graph $H$, namely $t(H;x,x)$---recall that the Tutte polynomial is a two variable polynomial--and the interlace polynomial $q(G;x)$ of a graph $G$ derived from the medial graph of $G$ (Theorem~\ref{thm:mosa1}).\ Such a relation led to the following three interesting consequences:
\mybreak
\begin{itemize}
\item[--] the \#P--completeness of the interlace polynomial of Arratia,
Bollob\'{a}s and Sorkin~\cite{abs} in the general case;
\mybreak
\item[--] a characterization of BDH graphs via the so-called \emph{$\gamma$ invariant}, (i.e., the coefficient of the linear term of the interlace polynomial);
\mybreak
\item[--] an effective proof that the interlace polynomial is polynomial-time computable within BDH graphs.
\end{itemize}
In view of a result due to Aigner and van der Holst (Theorem~\ref{thm:aighol}), the latter two consequences in the list above are straightforward consequences of Theorem~\ref{thm:main} (see Section~\ref{sec:stateart}).
\mybreak
Besides the new direct proofs of these results, Theorem~\ref{thm:main} has some more applications.
\mybreak
\textbf{1.} An easy proof that the class of BDH graphs form a class of graphs closed under \emph{edge--pivoting} (roughly, $GF(2)$-pivoting on the adjacency matrix). In other words, by pivoting (over $GF(2)$) on a nonzero entry of the \emph{reduced adjacency} matrix of a BDH graph, yields the \emph{reduced adjacency} matrix of another BDH graph.
\mybreak
\textbf{2.}~Syslo's characterization's of series--parallel graphs in
terms of \emph{Depth First Search} (DFS) \emph{trees}: the characterization asserts that a connected graph $H$ is series--parallel if and only if every spanning tree of $H$ is a DFS-tree of one of its 2--\emph{isomorphic} copies. In other words, up to 2--\emph{isomorphism}, series--parallel graphs have the characteristic property that their spanning trees can be oriented to become \emph{arborescences} so that the corresponding fundamental cycles become directed circuits (cycles whose arcs are oriented in the same way).\ Recall that an \emph{arborescence} is a directed tree with a single special node distinguished as the \emph{root} such that,
for each other vertex, there is a directed path from the root to that vertex.
\mybreak
\textbf{3.}~A characterization in terms of forbidden induced subgraphs of those BDH graphs whose reduced adjacency matrix generates the cycle matroid of a \emph{self--dual outer--planar graph}, namely, an outer--planar graph whose plane dual is also outer--planar. Remark that outer--planar graphs form a subclass of series--parallel graphs.\ The characterization asserts that such BDH graphs are precisely those that can be transformed into a \emph{bipartite chain graph} by a sequence of edge--pivoting. 
\mybreak
\textbf{4.}~New polynomially solvable instances for the problem of finding \emph{integer multi commodity flow} of maximum value: if the demand graph of a series--parallel graph is a co--tree, then the maximum value of a multi commodity flow equals the minimum value of a \emph{multi-terminal cut}; furthermore both a maximizing flow and a minimizing cut can be found in strongly polynomial time.
\mybreak
\textbf{Organization of the paper}
The rest of the paper is organized as follows: in Section~\ref{sec:prel} we give the basic notions used throughout the rest of the paper. In Section~\ref{sec:proofs} we prove our main result (Theorem~\ref{thm:main}) and discuss how it fits within \emph{circle graphs}. One more short proof is given in the appendix.\ Theorem~\ref{thm:main} is then applied in Section~\ref{sec:application}: in Section~\ref{sec:piv}, we deduce that BDH graphs is a nontrivial subclass of bipartite circle graphs closed under edge--pivoting. A direct proof of this fact would have been far from being technically trivial (compare with the proof of Theorem~\ref{thm:arrowt2closed}); in Section~\ref{sec:stateart}, we re-prove the previously mentioned couple of results in~\cite{mosa};\ in Section~\ref{sec:dfs} we re-prove Syslo's characterization of series--parallel graphs and give a sort of hierarchy of characterizations of 2--connected planar graphs by the properties of their spanning trees;\ in Section~\ref{sec:treeorbit}, we characterize self--dual outer--planar graphs and their fundamental graphs;\ finally in Section~\ref{sec:flow}, we give an application to Multi commodity--Flows in series--parallel graphs. 
\begin{figure}
    \begin{center}
             \includegraphics[width=13cm]{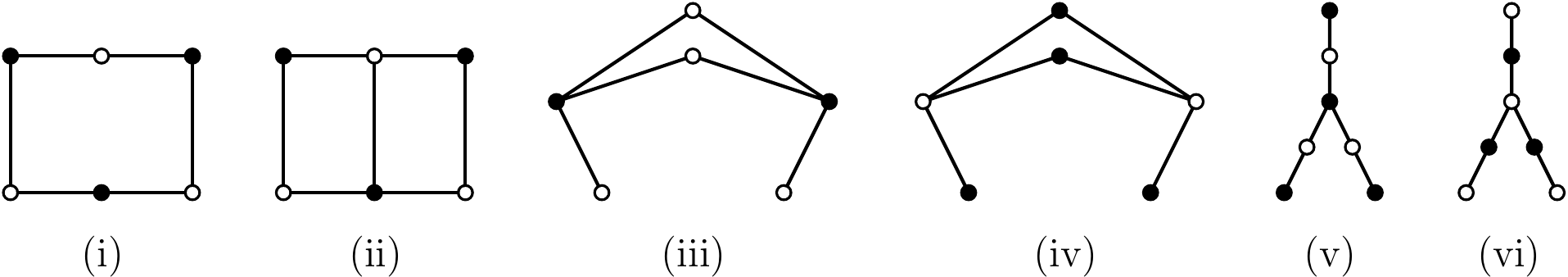}
    \end{center}
    \caption{Special graphs used in the paper: (i) is the hole $C_6$, (ii) is the domino $\boxminus$,\ (iii) and (iv) are copies of arrows while (v) and (vi) are copies of $T_2$; notice that for arrows and $T_2$  there is no automorphism that invertes the colors, namely, if $\sigma$ is an automorphism, then there is a vertex such that $\sigma(v)$ is in the same color class of $v$.}
    \protect\label{fi:special}
\noindent\hrulefill%
\end{figure}
\section{Preliminaries}\label{sec:prel}
For a graph $G$ the edge $e$ with endvertices $x$ and $y$ will be denoted by $xy$. The graph induced by $U\subseteq V(G)$ is
denoted by $G[U]$. The graph $G[U]$ will be also referred to as
the graph obtained by \emph{deleting} the vertices in
$V(G)\setminus U$. If $F\subseteq E(G)$, the graph $G-F$ is the graph $(V(G),E(G)-F)$.\ A \emph{cut--edge} of a graph $G$ is an edge whose removal disconnects the graph.

A \emph{digon} is a pair of parallel edges, namely a cycle with two edges. A~\emph{hole} in a
bipartite graph is an induced subgraph isomorphic to $C_n$ for
some $n\geq 6$. A \emph{domino} is a subgraph isomorphic to the
graph obtained from $C_6$ by joining two antipodal vertices by a
chord. The domino is denoted by $\boxminus$. An \emph{arrow} is the graph obtained from $C_4$ by adding two pending vertices adjacent to two vertices of the same color, while $T_2$ is the subdivision of the claw $K_{1,3}$ obtained by inserting a new vertex into every edge (see Figure~\ref{fi:special}). Graph $2K_2$ is the complement of $C_4$, namely it is a graph consisting of two non-adjacent edges.\ A bipartite graph with color classes $A$ and $B$ is a \emph{bipartite chain graph} if the members of both families $(N_G(u)\ | \ U\in A)$ and $(N_G(v)\ | \ v\in B)$ are inclusion-wise comparable. Equivalently, $G$ is a bipartite chain graph if it is bipartite and does not contain any induced copy of $2K_2$.\ Let $\mathcal{F}$ be a family of graphs. We say that $G$ is $\mathcal{F}$--free if $G$ does not contain any induced copy of a member of $\mathcal{F}$. If $G$ is $\mathcal{F}$--free and $\mathcal{F}=\{G_0\}$, then we say that $G$ is $G_0$--free.
\mybreak
The \emph{weak-dual} of a plane graph $H$ is the subgraph of the plane dual of $H$ induced by the vertices corresponding to the bounded faces.\ An \emph{outer-planar} graph is a planar graph that can be embedded in the plane so that all vertices are on the
outer face.\ Any such embedding is referred to as an \emph{outplane} embedding. Outer-planar graphs are characterized as those graphs not containing a minor isomorphic to  either $K_4$ or $K_{2,3}$. Equivalently, are those series--parallel graphs that do not contain a minor isomorphic (or a subdivision of) $K_{2,3}$. 2--connected outer-planar graphs posses a unique Hamiltonian cycle bounding the outer face.\ The weak-dual of an outplane embedding of a 2--connected outer-planar graph is a tree. 
\mybreak
Graph dealt with in this paper are, in general, not assumed to be vertex-labeled.\ However, when needed, vertices are labeled by the first $n$ naturals where $n$ is the order of $G$. We denote labeled and unlabeled graph with the same symbol.\ If $u$ and $v$ are two vertices of $G$, then a \emph{label swapping} at $u$ and $v$ (or simply $uv$-swapping) is the labeled graph obtained by interchanging the labels of $u$ and $v$. For a bipartite graph $G$ with color classes $A$ and $B$, let $\mathsf{A}\in \{0,1\}^{A\times B}$ be the \emph{reduced adjacency matrix of $G$}, namely, $\mathsf{A}$ is the matrix whose rows are indexed by the vertices of $A$, whose columns are indexed by the vertices of $B$ and where $\mathsf{A}_{u,v}=1$ if and only if $u$ and $v$ are adjacent vertices of $G$. The \emph{incidence graph} of a matrix $\mathsf{A}\in \{0,1\}^{A\times B}$ is the bipartite graph with color classes $A$ and $B$ and where $u\in A $ and $v\in B$ are adjacent if and only $a_{u,v}=1$.
\mybreak
We revise very briefly some basic notion in matroid theory~\cite{oxley,truemper,welsh}.
\begin{figure}
    \begin{center}
             \includegraphics[width=7cm]{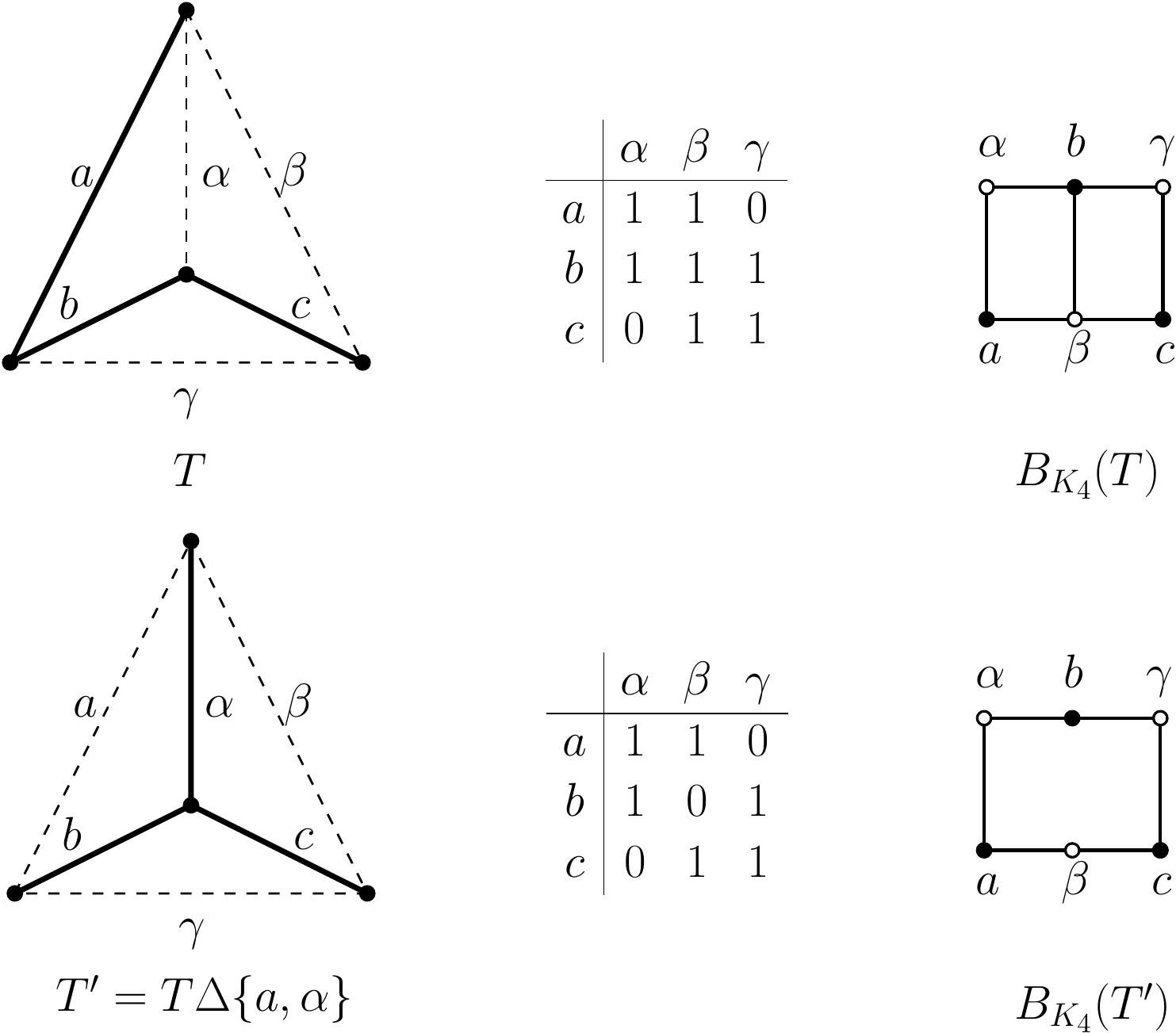}
    \end{center}
    \caption{Two fundamental graphs of $K_4$ with respect to two spanning tree $T$ and $T'$ along with the corresponding matrices and the respective fundamental graphs. The fundamental graphs with respect to $T'$ arises from the one with respect to tree by pivoting along edge $\alpha a$.}
    \protect\label{fi:fundgraph}
\noindent\hrulefill%
\end{figure}
For a $\{0,1\}$-matrix $\mathsf{A}$ the \emph{binary matroid generated by} $\mathsf{A}$, denoted by $M(\mathsf{A})$, is the matroid whose elements are the indices of the columns of $\mathsf{A}$ and whose independent sets are those subsets of elements whose corresponding columns are linearly independent over $GF(2)$. A \emph{binary matroid} is a matroid isomorphic to the binary matroid generated by some $\{0,1\}$-matrix $\mathsf{A}$. If $T$ is a basis of a binary matroid $M$ and $f\not\in T$, then $T\cup \{f\}$ contains a unique minimal non independent set $C(f,T)$, namely, if $F$ is a proper subset of $C(f,T)$, then $F$ is an independent set of $M$. Such a set $C(f,T)$ is the so called \emph{fundamental circuit
through} $f$ \emph{with respect to} $T$ and $C(f,T)-\{f\}$ is the corresponding \emph{fundamental path}. A \emph{partial representation} of a binary matroid $M$ is a $\{0,1\}$-matrix $\widetilde{\mathsf{A}}$ whose columns are the incidence vectors over the elements of a basis of the fundamental circuits with respect to that basis. A \emph{fundamental graph of a binary matroid} $M$ is simply the incidence bipartite graph of any of its partial representations. Therefore a bipartite graph $G$ is the fundamental graph of a binary matroid $M$ if $G$ is isomorphic to the graph $B_M(T)$ with color classes $T$ and $\overline{T}$ for some basis $T$ and co-basis $\overline{T}$ (i.e., the complement of a basis) of $M$ and where there is an edge between $e\in T$ and
$f\in\overline{T}$ if $e\in C(f,T)$, $C(f,T)$ being the fundamental circuit
through $f$ with respect to $T$. If $\widetilde{\mathsf{A}}$ is a partial representation of a binary matroid $M$, then $M\cong M([\,\mathsf{I}\,|\,\widetilde{\mathsf{A}}\,])$, that is $M$ is isomorphic to the matroid generated by $[\,\mathsf{I}\,|\,\widetilde{\mathsf{A}}\,]$. Clearly, $\widetilde{\mathsf{A}}$ is partial representation of $M$ with rows and columns indexed by the elements of the co-basis $\overline{T}$ and the basis $T$, respectively, if and only if $\widetilde{\mathsf{A}}$ is the reduced adjacency matrix of $B_T(M)$, where the color class $T$ indexes the columns of $\widetilde{\mathsf{A}}$.
\mybreak
The \emph{cycle matroid} (also known as \emph{graphic matroid)} of $H$, denoted by $M(H)$, is the matroid whose elements are the edges of $H$ and whose independent sets are the forests of $H$.\ If $H$ is connected, then the bases of $M(H)$ are precisely the spanning trees of $H$ and its co-bases are precisely the \emph{co-trees}, namely the subgraphs
spanned by the complement of the edge--set of a spanning tree. A matroid $M$ is a \emph{cycle matroid} if it isomorphic to the cycle matroid of some graph $H$. Cycle matroids are binary: if $M$ is a cycle matroid, then there is a graph $H$ and a spanning forest of $H$ such that $[\,\mathsf{I}\,|\,\widetilde{\mathsf{A}}\,]$ where $[\widetilde{\mathsf{A}}\,]$ is the $\{0,1\}$-matrix whose columns are the incidence vectors over the edges of a spanning forest of the fundamental cycles with respect to that spanning forest. A \emph{fundamental graph of a graph} $H$ is simply the fundamental graph of its cycle matroid $M(H)$. For a graph $H$ and one of its spanning forest $T$, we abridge the notation $B_{M(H)}(T)$ into $B_H(T)$ to denote the fundamental graph of $H$ with respect to $T$ (see~Figure~\ref{fi:fundgraph}).\ If $H$ is 2--connected, then $B_H(T)$ is connected.\ Moreover, $B_H(T)$
does not determine $H$ in the sense that non-isomorphic graphs may
have isomorphic fundamental graphs.\ This because, while it is certainly true that isomorphic graphs have isomorphic cycle matroids, the converse is not generally true (See Figure~\ref{fi:2_iso}).\ Two graphs having isomorphic cycle matroids are called 2--\emph{isomorphic}. 
\begin{figure}[H]
    \begin{center}
             \includegraphics[width=9cm]{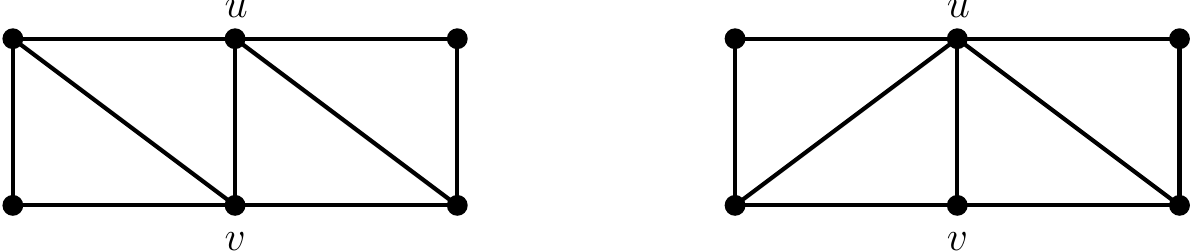}
    \end{center}
    \caption{Two 2--isomorphic graphs that are not isomorphic.}
    \protect\label{fi:2_iso}
\noindent\hrulefill%
\end{figure}

\section{BDH graphs are fundamental graphs of series parallel graphs}\label{sec:proofs}
In this section we prove our main result.
\begin{theorem}\label{thm:main} A connected bipartite graph $G$ is a bipartite distance hereditary graph if and only if it is a fundamental graph of a 2--connected series--parallel graph.
\end{theorem}
\begin{proof}Recall that if $G$ is a bipartite graph, then $M^G$ denotes the binary matroid generated by the reduced adjacency matrix of $G$. Let us examine preliminarily the effect induced on a fundamental graph $B_H(T)$ of a 2--connected graph $H$ by series and parallel extension and, conversely (and in a sense ``dually''), the effect induced on $M^G$ by extending a connected bipartite graph $G$ through the addition of pending vertices and twins.\ If $M^G$ is a graphic matroid and $H$ is one of the 2--isomorphic graphs whose cycle matroid is isomorphic to $M^G$, then the following table summarizes these effects.
\begin{table}[h]\label{table:spbdh}
	\centering
	\begin{tabular}{lll}
		Operation on $H$ & & Operation on $B_H(T)$\\ \hline
		Parallel extension on edge $a$ of $T$ &$\leftrightarrow$& adding a pending vertex in color class $\overline{T}$ adjacent to $a$ \\
		Series extension on edge $a$ of $T$ & $\leftrightarrow$ & adding a twin of $a$ in color class $T$\\
		Parallel extension on edge $b$ of $\overline{T}$ & $\leftrightarrow$ & adding a twin
		of $b$ in color class $\overline{T}$\\ Series extension on edge $b$ of $\overline{T}$ & $\leftrightarrow$ &
		adding a pending vertex in color class $T$ adjacent to $b$.\\ \hline
	\end{tabular}
	\caption{\small The effects of series and parallel extension on $H$ on
		its fundamental graph $B_H(T)$.}
	\protect
	\noindent\hrulefill
\end{table}
\mybreak
We can now proceed with the proof.\ The \emph{if part} is proved by induction on the order of $G$.\ The
assertion is true when $G$ has two vertices because $K_2$ is a BDH
graph and at the same time is also the fundamental graph of a
digon. Therefore $M^G$ is a cycle matroid.\ Let now $G$ have $n\geq 3$ vertices and assume that the
assertion is true for BDH graphs with $n-1$ vertices. By Bandelt
and Mulder's construction $G$ is obtained from a BDH graph $G'$ either by
adding a pending vertex or a twin. Let $H'$ be a series--parallel
graph having $G'$ as fundamental graph with respect to some
spanning tree. Since, by Table 1, the latter two operations
correspond to series or parallel extension of $H'$, the result
follows by Duffin's construction of series--parallel graphs.\ Conversely, let $G$
be the fundamental graph of a series--parallel graph $H$ with
respect to some tree $T$.\ By Duffin's construction of series--parallel graphs and
Table 1, $G$ can be constructed starting from a single edge by
either adding twins or pending vertices. Therefore, $G$ is a BDH
graph by Bandelt and Mulder's construction.
\end{proof}
There are also other ways to obtain the result: one is briefly addressed in Section~\ref{sec:piv}. Another one builds on the results of~\cite{abs,mosa}. For the interested reader, we give such a proof in the appendix.
Before going through applications, let us discuss how Theorem~\ref{thm:main} relates to \emph{circle graphs}, a thoroughly investigated class of graphs which we now briefly describe.

A \emph{double occurrence word} $\mathbf{w}$ over a finite alphabet $\Sigma$ is a word in which
each letter appears exactly twice---$\mathbf{w}$ is cyclic word, namely, it is the equivalence class of a linear word modulo cyclic shifting and reversal of the orientation. Two distinct symbols of $\Sigma$ in $\mathbf{w}$ are \emph{interlaced} if one appears precisely once between the two occurrences of the other. By wrapping $\mathbf{w}$ along a circle and by joining the two occurrences of the same symbol of $\mathbf{w}$ by a chord labeled by the same symbols whose occurrences it joins, one obtains a pair $(S,\mathcal{C})$ consisting of a circle $S$ and a set $\mathcal{C}$ of chords of $S$. In knot theory terminology, such a pair is usually called a \emph{chord diagram} and its intersection graph, namely the graph whose vertex set is $\mathcal{C}$ and where two vertices are adjacent if and only if the corresponding chords intersects, is called \emph{the interlacement graph of the chord diagram} or the \emph{interlacement graph of the double occurrence word}.
\mybreak
A graph is an \emph{interlacement graph} if it is the interlacement graph of some chord diagram or of some double occurrence words. Interlacement graphs are probably better known as \emph{circle graphs}---the name \emph{interlacement graph} comes historically from the \emph{Gauss Realization Problem of double occurrence words} \cite{fom,rose,shtraldizu}--.
\mybreak
Distance hereditary graphs are circle graphs~\cite{bralespi}. Thus BDH graphs form a proper subclass of bipartite circle graphs.\ de Fraysseix \cite{defrai,defrai1} proved the following.
\begin{theorem}[de Fraysseix]\label{thm:defra}
A bipartite graph is a bipartite circle graph if and only if it is the fundamental graph of a planar graph.
\end{theorem}
Therefore Theorem~\ref{thm:main} specializes de Fraissex's Theorem to the subclass of series--parallel graphs.   

\section{Applications}\label{sec:application}
\subsection{BDH graphs are preserved by edge--pivoting}\label{sec:piv}
It follows from Theorem~\ref{thm:main} that with every 2--isomorphism class of 2--connected series--parallel graphs one can associate all the BDH graphs that are fundamental graphs of each member in the class. Therefore BDH graphs that correspond to the same 2--isomorphism class are graphs in the same ``orbit''. In this section we make precise the latter sentence and draw the graph-theoretical consequences of this fact.
\mybreak
\mybreak
Given a $\{0,1\}$-matrix $\mathsf{A}$, pivoting $\A$
over $GF(2)$ on a nonzero entry (the pivot element) means
replacing
$$\A=\left(\begin{array}{cc}
\mathsf{1} & \mv{a}\\
\mv{b} & \D
\end{array}
\right) \quad {\rm by} \quad \tilde{A}=\left(\begin{array}{cc}
\mathsf{1} & \mv{a}\\
\mv{b} & \mv{D+ba}
\end{array}
\right)
$$
where $\mv{a}$ is a row vector, $\mv{b}$ is a column vector, $\D$ is a
submatrix of $\A$ and the rows and columns of $\A$ have been
permuted so that the pivot element is $a_{1,1}$
(\cite{cornuejols}, p. 69, \cite{schrij}, p. 280).
If $\A$ is the partial representation of the cycle matroid of a graph $H$ (or more generally a binary matroid), then pivoting on a non zero entry, $C(e,f)$, say, yields a new tree (basis) with $f$ in the tree (basis) and $e$ in the co-tree (co-basis) and the matrix obtained after pivoting is a new partial representation matrix of the same matroid. Clearly the fundamental graphs associated with the two bases change accordingly so that pivoting on $\{0,1\}$-matrices induces an operation on bipartite graphs whose concrete interpretation is a change of basis in the associated binary matroid. The latter operation on bipartite graph will be still referred to as \emph{edge--pivoting} or simply to \emph{pivoting} in analogy to what happens for matrices (see also ~Figure~\ref{fi:fundgraph}). In the context of circle graphs, the operation of pivoting is a specialization to bipartite graph of the so called \emph{edge--local complementation}.\ Since any bipartite graph is a fundamental graph of some binary matroid, the operation of \emph{pivoting} can be described more abstractly as follows.
\mybreak
Given a bipartite graph with color classes $A$ and $B$ \emph{pivoting on edge $uv\in E(G)$} is the operation that takes $G$ into the graph $G^{uv}$ on the same vertex set of $G$ obtained by complementing the edges between $N_G(u)\setminus \{u\}$ and $N_G(v)\setminus \{v\}$ and then by swapping the labels of $u$ and $v$ (if $G$ is labeled). More formally, if $\ell_G: V(G)\rightarrow \mathbb{N}$ is a labeling of the vertices of $G$, then 
$$G^{uv}=(V(G),E(G)\Delta ((N_G(u)\setminus \{u\})\times (N_G(v)\setminus \{v\})))$$
and $\ell_{G^{uv}}$ is defined by $\ell_{G^{uv}}(u)=\ell_G(v)$, $\ell_{G^{uv}}(v)=\ell_G(u)$ and $\ell_{G^{uv}}(w)=\ell_G(w)$ if $w\not\in\{u,v\}$. If $e\in E(G)$ has endpoints $uv$, then we use $G^e$ in place of $G^{uv}$.
\mybreak
We say that a graph $\widetilde{G}$ is \emph{pivot-equivalent} to a graph $G$, written $\widetilde{G}\sim G$
if for for some $k\in \mathbb{N}$, there is a sequence
$G_1,\cdots,G_k$ of graphs such that $G_1\cong G$, $G_k\cong \widetilde{G}$ and, for
$i=1,\ldots, k-1$, $G_{i+1}\cong G_i^{e_i}$, $e_i\in E(G_i)$. The \emph{orbit of $G$}, denoted by $[G]$, consists of all graph that pivot-equivalent to $G$. 
\mybreak
For later reference, we state as a lemma the easy though important facts discussed above.\ Figure~\ref{fi:fundgraph} illustrates the contents of the lemma. 
\begin{lemma}\label{lem:matorbit}
Let $M$ be a connected graphic matroid. Then $M$ determines both a class $[G]$ of pivot-equivalent graphs and a class $[H]$ of 2--isomorphic graphs. In particular, any graph in $[G]$ is the fundamental graph of some 2--isomorphic copy of $H$ and the fundamental graph of any graph in $[H]$ is pivot-equivalent to $G$. 
\end{lemma}
The operations of pivoting and of taking induced subgraphs commute in (bipartite) graphs. 
\begin{lemma}[see \cite{abs}]\label{lem:0}
Let $G$ a bipartite graph, $U\subseteq V(G)$. Then $G^e[U]\cong (G[U])^e$.
\end{lemma}
The next lemma relates in the natural way minors of a cycle matroid (graph) to the induced subgraphs of the fundamental graphs associated with the matroid (graph).
\begin{lemma}\label{lem:matroidorbit}
Let $M$ and $N$ be cycle matroids. Let $G$ be any of the fundamental graphs of $M$ and let $K$ be any of the fundamental graphs of $N$. Then $N$ is a minor of $M$ if and only if $K$ is an induced subgraph in some bipartite graph in the orbit of $G$. Equivalently, $N$ is a minor of $M$ if and only if $G$ contains some induced copy of a graph in the orbit of $K$.
\end{lemma}
To get acquainted with pivoting, the reader may check Lemma~\ref{lem:1} with the help of Figure~\ref{fi:lemma}. Refer to Section~\ref{sec:prel} for the definition of \emph{domino} and \emph{hole} and to Figure~\ref{fi:special} for an illustration.
\begin{figure}[h]
	\begin{center}
		\includegraphics[width=12.5cm]{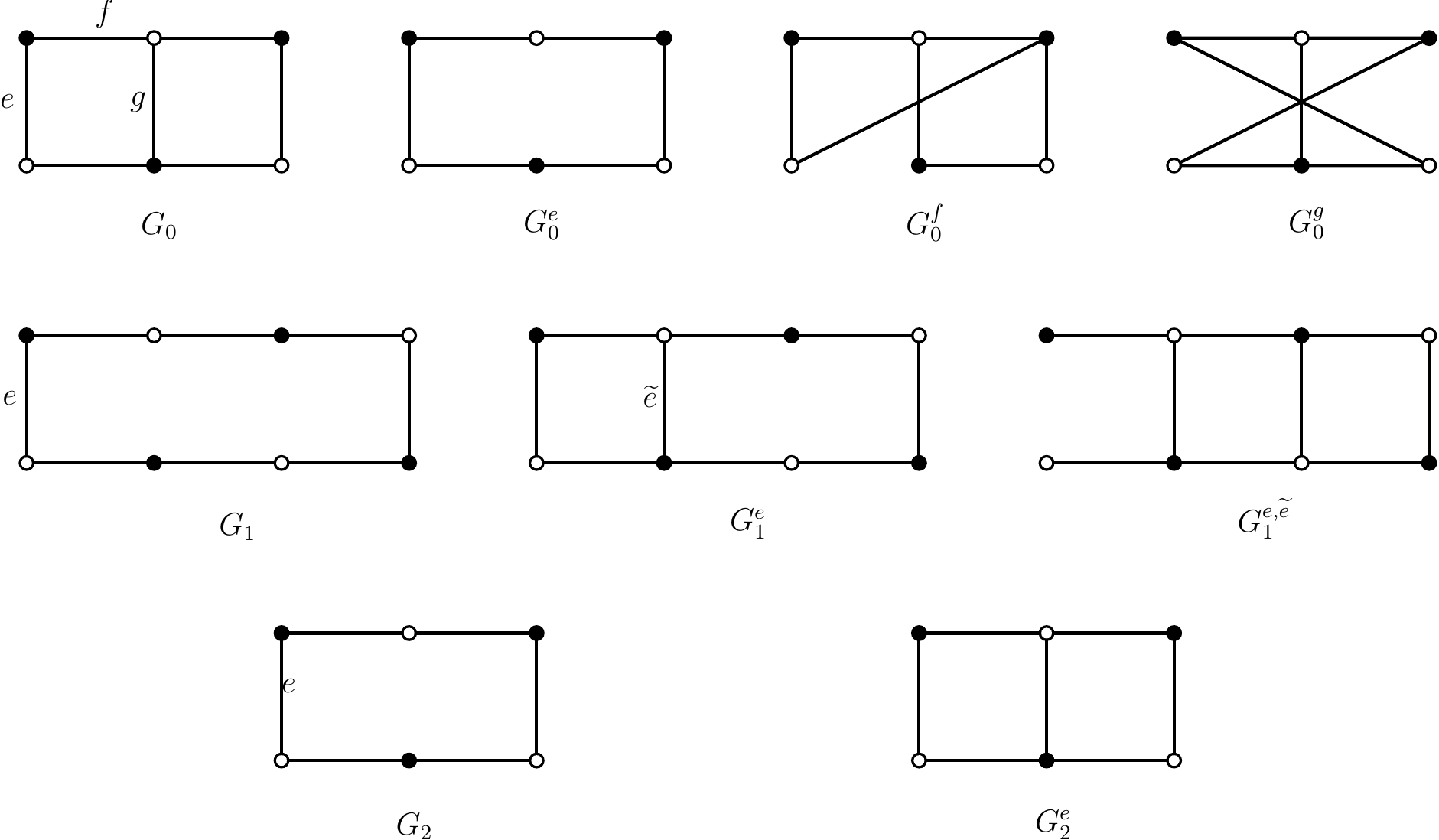}
	\end{center}
	\caption{The effect of pivoting a graph $G$ along some of its edges when $G\cong \boxminus, C_8, C_6$.}
	\protect\label{fi:lemma}
	\noindent\hrulefill%
\end{figure}
\begin{lemma}\label{lem:1}~~
	\begin{itemize}
		\item[--] If either $H\cong \boxminus$
		or $H\cong C_k$, $k\geq 6$, then for each $uv\in E(H)$ there exists an induced subgraph
		$H'$ of $H^{uv}$ such that either $H'\cong \boxminus$
		or $H'\cong C_k$, $k\geq 6$. 
		\item[--] If $G\cong C_k$, $k\geq 6$, then there is a graph $\widetilde{G}$ in the orbit of $H$ such that $\widetilde{G}$ contains an induced copy of either $\boxminus$ or $C_6$.	
	\end{itemize}
\end{lemma}
\begin{proof} By inspecting Figure~\ref{fi:lemma} one checks that if $G\cong \boxminus$, then either $G^e\cong \boxminus$ or $G^e\cong C_6$. If $G\cong C_6$, then $G^e\cong \boxminus$
	for every $e\in E(G)$. If $G\cong C_k$, $k>6$, then by pivoting on
	$uv\in E(G)$  and deleting $u$ and $v$ results in a graph $G'\cong
	C_{k-2}$. In particular, by repeatedly pivoting on new formed edges (like edge $\widetilde{e}$ of graph $G_1^e$ in Figure~\ref{fi:lemma}), one obtains a graph in the orbit of $G$ which contains an induced copy of either $\boxminus$ or $C_6$. The second part of the proof is left to reader.
\end{proof}
We are ready to extract the graph-theoretical consequence of Theorem~\ref{thm:main}. To this end let us recall that besides their constructive characterization, Bandelt and Mulder characterized the class of BDH graphs also by forbidden induced subgraph as follows.
\begin{theorem}[Bandelt and Mulder \cite{bm}, Corollaries 3 and 4]\label{thm:bdhchar} Let $G$ be a connected bipartite graph. Then $G$ is BDH if and only if $G$ contains neither holes nor induced dominoes.
\end{theorem}
	The following two corollaries, follow straightforwardly from Theorem~\ref{thm:main} after Theorem~\ref{thm:bdhchar} and assert that the class of BDH graph, that is $\{hole,domino\}$--free graphs, is closed under pivoting, namely, that the orbit of a $\{hole,domino\}$--free graph consists of $\{ hole,domino\}$--free graphs. By Lemma~\ref{lem:1}, we see that holes contain induced dominoes in their orbits and, on the other hand, $C_6$ is in domino's orbit. However this does not prove that pivoting preserves the property of being $\{hole,domino\}$--free.\ A direct proof of this fact would have been rather involved and technical requiring a complicate case analysis.\ Compare for a flavor with the proofs of Theorem~\ref{thm:arrowt2closed} in Section~\ref{sec:treeorbit} asserting that the class of $\{arrow,T_2\}$--free BDH graphs is closed under pivoting. 
\begin{coro}\label{cor:main} The following statement about a chordal bipartite
graph $G$ are equivalent:
\begin{enumerate}[{\rm (i)}]
\item $G$ does not contain any induced domino;
\item any graph in the orbit of $G$ is a chordal bipartite graph.
\end{enumerate}
\end{coro}

\begin{coro}\label{cor:main1} Let $G$ be a bipartite domino-free
graph. If $G$ is chordal, then so is any other graph in its orbit.
\end{coro}
\subsection{BDH graphs and the interlace polynomial}\label{sec:stateart}
As already mentioned, Ellis-Monaghan and Sarmiento related series--parallel graphs and BDH graphs topologically, via the medial graph.\ Let $H$ be a plane (or even a 2-cell embedded graph in an oriented surface). For our purposes, we can assume that $H$ is 2--connected. The medial graph $m(H)$ of $H$ is the graph obtained as follows: first place a vertex $v_e$ into
the interior of every edge $e$ of $H$.\ Then, for each face $F$ of $H$, join $v_e$ to $v_f$ by an edge
lying in $F$ if and only if the edges $e$ and $f$ are consecutive on the boundary of $F$. Remark that if $F$ is bounded by a digon $\{e,e'\}$ then $e$ and $e'$ are consecutive twice. Let $m(H)$ be the plane (2-cell embedded) graph obtained in this way. The graph underlying $m(H)$ is the medial graph of $H$. The medial graph is clearly 4-regular, as each face creates two adjacencies for each edge on its boundary. Moreover, it can be oriented so that it becomes a 2-in,2-out graph.\ Given a 4-regular labeled graph $N$ and one of its Eulerian circuit $C$, we 
can associate with $N$ a double occurrence word $\mathbf{w}$ which is the word of the labels of the vertices cyclically met during the tour on $C$. The circle graph induced by $\mathbf{w}$ is called the \emph{the circle graph of $N$}.\ Ellis-Monaghan and Sarmiento, Building also on the relations between \emph{Martin polynomial} and the symmetric \emph{Tutte Polynomial}, proved the following relation between the symmetric Tutte polynomial $t(H;x,x)$ of a planar graph $H$ and the vertex nullity interlace polynomial $q(G;x)$ of a graph $G$ derived, as described in the theorem below, from the medial graph of any of its plane embedding.
\begin{theorem}[[Monhagan, Sarmiento, 2007 \cite{mosa}]\label{thm:mosa1}
	If $H$ is a plane embedding of a planar graph and $G$ is the circle graph of some
	Eulerian circuit of the medial graph of $H$, then $q(G;x)=t(H;x,x)$.
\end{theorem}
The results was then specialized so as to give following characterization of BDH graphs.
\begin{theorem}[[Monhagan, Sarmiento, 2007 \cite{mosa}]\label{thm:mosabis}
	$G$ is a BDH graph with at least two vertices if and only if
	it is the circle graph of an Euler circuit in the medial graph of a plane embedding of a series--parallel
	graph $H$.
\end{theorem}
Using Theorem~\ref{thm:mosa1} and Theorem~\ref{thm:mosabis}, the authors deduced the following consequences stated below as Corollary~\ref{coro:mosa0}, Corollary~\ref{coro:mosai} and Corollary~\ref{coro:mosaii}.
\begin{coro}\label{coro:mosa0}
	The vertex-nullity interlace polynomial is \#P-hard in general.
\end{coro}

\begin{coro}\label{coro:mosai}
	If $G$ is a BDH graph, then $q(G;x)$ is polynomial-time computable.
\end{coro}
Corollary~\ref{coro:mosai} follows because the Tutte polynomial is polynomial-time computable for series--parallel graphs~\cite{OW}. 
\begin{coro}\label{coro:mosaii}
	A connected graph $G$ is a BDH graph if and only if the coefficient of the linear term of $q(G;x)$ equals 2.
\end{coro}
The latter coefficient referred to in Corollary~\ref{coro:mosaii}, denoted by $\gamma(G)$, is called the $\gamma$-invariant of $G$ in analogy with the Crapo invariant $\beta(G)$ which is the common value of the coefficients of the linear terms of $t(G;x,y)$ where $G$ has at least two edges. By a result due to Brylawski \cite{bry} (in the more general context of
matroids) series--parallel graphs can be characterized by the value of the Crapo invariant as follows: a graph $G$ is a series--parallel graph if and only if $\beta(G)=1$. Both the corollaries above can be deduced directly by Theorem~\ref{thm:main} after the following result due to Aigner and van der Holst \cite{aighol}.  
\begin{theorem}[Aigner, van der Holst, 2004 \cite{aighol}]\label{thm:aighol}
	If $G$ is a bipartite graph, then
	$$q(G; x)= t(M^G;x, x)$$
	where $M^G$ is the binary matroid generated by the reduced adjacency matrix of $G$ and $t(M^G;x, x)$ is the Tutte polynomial of $M^G$.
\end{theorem}
Theorem~\ref{thm:main} and Theorem~\ref{thm:aighol} have the following straightforward consequence which re-proves directly Corollary~\ref{coro:mosai} and Corollary~\ref{coro:mosaii}.
\begin{coro}\label{coro:main}
	If $G$ is BDH graph, then
	$$q(G; x)= t(H;x, x)$$
	for some series--parallel graph $H$ having $G$ as fundamental graph and where $t(H;x, x)$ is the Tutte polynomial of $H$, namely the Tutte polynomial of the cycle matroid of $H$.
\end{coro}

\subsection{Characterizing series--parallel graphs by DFS-trees}\label{sec:dfs}
As credited by Syslo \cite{syslo}, Shinoda, Chen, Yasuda,
Kajitani, and Mayeda, proved that series--parallel graphs can be
completely characterized as in Theorem~\ref{thm:syslo} by a property of their spanning trees, and Syslo himself gave a constructive algorithmic proof of the result~\cite{syslo}. 
\begin{theorem}[S. Shinoda et al., 1981; Syslo, 1984]\label{thm:syslo}
Every spanning tree of a connected graph $H$ is a DFS-tree of one of its 2--isomorphic copies if and only if $H$ is a series--parallel graph.
\end{theorem}
When $H$ is assumed to be 2--connected (an assumption that guarantees the connectedness of its fundamental graphs), Theorem~\ref{thm:syslo} can be equivalently stated as follows. 
\mybreak
Let $\mathcal{T}$ be a family of (possibly oriented) trees and let $G$ be a bipartite graph with color classes $A$ and $B$. We say that $G$ is a \emph{path/~$\mathcal{T}$ bipartite graph on $A$} if there exists a member $T$ of $\mathcal{T}$ and a bijection $\xi: A\rightarrow E(T)$ such that, for each $b\in B$, $\{\xi a \ |\ a\in \{b\}\cup N_G(b)\}$ is the edge--set (arc--set if $T$ is oriented) of a simple cycle (directed circuit if $T$ is oriented) in the (oriented) graph $(V(T), A\cup B)$. \emph{Path/~$\mathcal{T}$ bipartite graphs on $B$} are defined similarly.\ $G$ is a \emph{path/~$\mathcal{T}$ bipartite graphs} if it is a path/~$\mathcal{T}$ bipartite graph on $A$ or on $B$.\ $G$ is a \emph{self--dual path/~$\mathcal{T}$} bipartite graphs if it is a path/~$\mathcal{T}$ bipartite graph on both $A$ and $B$.\ In any case $T$ will be referred to as a \emph{supporting tree} for $G$.
\mybreak
Recall that an \emph{arborescence} is a directed tree with a single special node distinguished as the \emph{root} such that,
for each other vertex, there is a directed path from the root to that vertex. A \emph{DFS tree} for a connected graph $H$ (in the sense of~\cite{syslo}), is a pair $(T,\phi)$ consisting of a spanning tree and an orientation of $H$, such that $\phi T$ is a spanning arborescence of $\phi H$ and for each $f\in E(H)\setminus E(T)$, $\phi C(f,T)$ is a directed circuit in $\phi H$ (i.e, all arcs of $\phi C(f,T)$ are oriented in the same way). By choosing for $\mathcal{T}$ the class \textbf{arborescence} of arborescences, one can reformulate Theorem~\ref{thm:syslo} in the following way
\cadre
\label{st:2} $H$ is series--parallel if and only if for each spanning tree $T$ of $H$ the fundamental graph $B_T(H)$ is a self--dual path/\textbf{arborescence} bipartite graph.
\endcadre
\mybreak
Indeed, if $(T,\phi)$ is a DFS-tree in a 2--isomorphic copy $H'$ of $H$, then $T$ is a spanning tree of graph $H'$ whose cycle matroid is $M(H)$;\ hence $B_H(T)\cong B_{H'}(T)$ and $\phi T$ is a supporting arborescence for $B_H(T)$.\ Conversely, suppose that $G$ is a fundamental graph of $H$ and that $G$ is a path/\textbf{arborescence} bipartite graph. Let $G$ have color classes $A$ and $B$. Since $G$ is a path/\textbf{arborescence} bipartite graph, then there is a supporting arborescence $\overrightarrow{T}$ for $G$ that induces an orientation $\phi$ of the graph $H'=(V(T), A\cup B)$, $T$ being the underlying undirected graph of $\overrightarrow{T}$. Clearly $(T,\phi)$ is a DFS tree in $H'$ which in turn is 2--isomorphic to $H$ because $G$ is one of its fundamental graphs (i.e., $H$ and $H'$ have the same cycle matroid). 
\mybreak
Statement~\eqref{st:2} is now a rather straightforward consequence of Corollary~\ref{cor:main} and the fact that BDH graphs are self--dual path/\textbf{arborescence} bipartite graphs as shown by the following result.
\begin{theorem}\label{thm:paolo} Every connected BDH is a self--dual path/\textbf{arborescence} bipartite graph. 
\end{theorem}
\begin{proof}
Let $G$ be a connected BDH and let $\sigma=(v_1,v_2\ldots,v_n)$ be one of the ordered sequence of pending and twins in the Bandelt Mulder's construction of $G$. Call any such sequence a \emph{defining sequence}. It suffices to prove that, for any defining sequence of $G$, $G$ is a path/\textbf{arborescence} bipartite graph on the color class containing $v_1$. Indeed since if $\sigma=(v_1,v_2\ldots,v_n)$ is a defining sequence, then $\sigma'=(v_2,v_1\ldots,v_n)$ is still a defining sequence, we conclude that $G$ is a path/\textbf{arborescence} bipartite graph on both the color classes as stated. Let $A$ be the color class of $G$ containing $v_1$. Hence the other color class $B$ contains $v_2$.\ To prove that $G$ is a path/\textbf{arborescence} bipartite graph on $A$ one has to exhibit a tree $T$ on the edge--set $A$ and an orientation $\phi$ of $T$ such that $\left(\phi a \ |\ a\in N_G(b)\right)$ is the arc set of a directed path in $\phi T$. The proof is by induction on the order of $G$. If $G$ has order 2, then $G\cong K_2$, and we are done because it suffices to choose $T$ as a single edge and to orient it arbitrarily to support $G$. Suppose that every BDH graph on less than $n$ vertices is a path/\textbf{arborescence} bipartite graph on the color class containing the first vertex of one its defining sequence and let $G_{n-1}$ be the subgraph of $G$ induced by the first $n-1$ vertices of $\sigma$. Furthermore let $A_{n-1}=A-\{v_n\}$ and $B_{n-1}=B-\{v_n\}$. Clearly $A_{n-1}$ may coincide with $A$. By the inductive hypothesis, $G_{n-1}$ is a path/\textbf{arborescence} bipartite graph on $A_{n-1}$. Therefore there exists a tree $T_{n-1}$ with edge set $A_{n-1}$ and an orientation $\phi_{n-1}$ such that $\phi_{n-1} T_{n-1}$ is a supporting arborescence for $G_{n-1}$. Only the following three cases can occur.
\mybreak
\textbf{--$v_n\in A$ and $v_n$ is a pending vertex.} Let $w\in B_{n-1}$ be the unique neighbor of $v_n$ and let $\phi_{n-1} P_w$ be the path spanned by $N_{G_{n-1}}(w)$ in $T_{n-1}$. Possibly by reversing $\phi_{n-1}$ we may suppose that it orients out of the root of $T_{n_1}$. Let $z$ be the last vertex of $P_w$. Join a vertex $z'$ to $z$ and define $T=(V(T_{n-1})\cup\{z'\},E(T_{n-1})\cup\{zz'\}$ and $\phi$ by $\phi(zz')=(z,z')$ and $\phi(e)=\phi_{n-1}(e)$ $\forall e\in E(T_{n-1})$. In other words, $zz'$ is added as the last edge of the en-longed path. 
\mybreak
\textbf{--$v_n\in A$ and $v_n$ is a twin vertex.} Let $w$ be any twin of $v_n$. Let $w$ have endpoints $\alpha$ and $\beta$ in $T_{n-1}$. Subdivide $w$ by the insertion of a new vertex $z$ and keep the orientation. More formally define $T=(V(T_{n-1})\cup\{z\},E(T_{n-1})\cup\{\alpha z,z\beta\}-\{w\})$ and $\phi$ by $\phi(\alpha z)=(\alpha,z)$, $\phi(z\beta)=(z,\beta)$, and $\phi(e)=\phi_{n-1}(e)$ $\forall e\in E(T_{n-1})-\{w\}$. 
\mybreak
\textbf{--$v_n\in B$.} Let $P=N_{G}(v_n)$. Either $P$ is single edge of $T_{n-1}$ (if $v_n$ is a pending vertex) or $P$ is a copy of the edge--set of an existing path of $T_{n-1}$ (if $v_n$ is a twin vertex). In either cases, if $\phi_{n-1} T_{n-1}$ supports $G_{n-1}$, then it supports $G$: just set $T=T_{n-1}$ and $\phi\equiv\phi_{n-1}$.
\mybreak
Since in any of the cases above $\phi T$ is a supporting arborescence for $G$, the proof is completed.
\end{proof}

\noindent\textbf{Proof of~\eqref{st:2}.} Let $H$ be a 2--connected series--parallel graph.\ Then, by Theorem~\ref{thm:main} $B_H(T)$ is BDH for each spanning tree $T$ of $H$. Hence, for every spanning tree $T$ of $H$, $B_H(T)$ is a self--dual path/\textbf{arborescence} bipartite graph by Theorem~\ref{thm:paolo}. Conversely, suppose that for every spanning tree $T$ of a 2--connected graph $H$, the fundamental graph $B_H(T)$ is a path/\textbf{arborescence} bipartite graph. Thus $B_H(T)$ is chordal (see, e.g., \cite{bralespi}). Moreover, since if $T'$ is any other spanning tree of $H$, then $B_H(T')$ is in the orbit of $B_H(T)$, we conclude that each bipartite graph in the orbit of $B_H(T)$ is a chordal bipartite graph. Therefore $B_H(T)$ is a BDH graph by Corollary~\ref{cor:main} and, consequently, $H$ is a series--parallel graph.\hspace*{\fill}$\Box$
\mybreak
It is worth observing that, in the same way as Theorem~\ref{thm:main} specializes de Fraysseix's Theorem~\ref{thm:defra}, Statement~\eqref{st:2} specializes the following statement (see also \cite{defrai}).
\cadre a bipartite graph is a bipartite circle graph if and only if it is a self--dual path/\textbf{tree} bipartite graph, \textbf{tree} being the class of trees.\endcadre
\begin{proof} By Whitney's planarity criterion~\cite{whi} a graph is planar if and only if its cycle matroid is also \emph{co-graphic}, namely, it is the dual matroid of another cycle matroid. Let now $G$ be the fundamental graph of a 2--connected graph $H$ with respect to some spanning tree $T$ of $H$. Let $\mathsf{A}$ be the reduced adjacency matrix of $G$ with columns indexed by the edges of $T$ and rows indexed by the edges of $\overline{T}$. Then, while $[\,\mathsf{I}\,|\,\mathsf{A}\,]$ generates $M(H)$, $[\,\mathsf{I}\,|\,\mathsf{A}^t\,]$ generates $M^*(H)$ the dual of $M(H)$. Hence, when $H$ is planar, by Whitney's planarity criterion, $M^*(H)$ is the cycle matroid of a 2--isomorphic copy of a plane dual $H^*$ of $H$. Therefore the neighbors of each vertex in the color class $T$ spans a path in the co-tree $\overline{T}$ which is in turn a spanning tree of a 2--isomorphic copy of plane dual $H^*$ of $H$. 
\end{proof}
In view of such a discussion it is reasonable to wonder whether there is a class of self dual path/~$\mathcal{T}_0$ class of bipartite graph closed under edge--pivoting, where $\mathcal{T}_0$ is a family of trees sandwiched between \textbf{trees} and \textbf{arborescences}.\ The next result gives a negative answer in a sense. In what follows \textbf{di-tree} is the class of oriented trees.
\begin{theorem}\label{thm:ept}
If $G$ be a connected bipartite graph whose orbit consists of self--dual path/\textbf{di-tree} bipartite graphs, then the orbit of $G$ consists of path/\textbf{arborescence} bipartite graphs.
\end{theorem}
\begin{proof}
Path/\textbf{di-tree} bipartite graphs are \emph{balanced} (see \cite{apollonio}). Recall that a bipartite graph $\Gamma$ is \emph{balanced} if its reduced adjacency matrix does not contain the vertex-edge adjacency matrix of a chordless cycle of odd order. Equivalently, $\Gamma$ is \emph{balanced} if each hole of $\Gamma$ has order congruent to zero modulo 4. Hence, since $G$ and any other graph in its orbit, is a self--dual path/\textbf{di-tree} bipartite graph, then $G$ and any other graph in its orbit must be balanced as well. Let $\widetilde{G}$ be any member of $[G]$ and suppose that $\widetilde{G}$ contains a hole $C$. Let $e\in E(C)$. The order $t$ of $C$ is at least eight, because $\widetilde{G}$ is balanced. Nevertheless $\widetilde{G}^e$ contains a hole of order $t-2$ by Lemma~\ref{lem:1}. Since $t-2\equiv 2\imod{4}$ we conclude that any graph in the orbit of $G$ must be hole-free. Therefore $G$ is BDH by Corollary~\ref{cor:main1}, and, by Theorem~\ref{thm:main}, it is the fundamental graph of a series--parallel graph. The thesis now follows by Statement~\eqref{st:2}.
\end{proof}

\begin{remark}
It is worth observing that by the proof above, if $\mathbf{A}$ is a class of balanced matrices closed under pivoting over $GF(2)$, then $\mathbf{A}$ consists of totally balanced matrices, namely those matrices whose bipartite incidence graph is hole-free. Actually, and more sharply, in view of Corollary~\ref{cor:main1}, every member of $\mathbf{A}$ is the incidence matrix of a $\gamma$-acyclic hypergraph~\cite{iwocapaper}.
\end{remark}

\subsection{Self--dual outer--planar graphs}\label{sec:treeorbit}
Series--parallel graphs form a self--dual class of planar graphs: any plane dual of a series--parallel graph is still series--parallel. It is natural to wonder whether there are subclasses of series--parallel graphs and, accordingly, subclasses of BDH graphs, that display the same self--duality. Outer-planar graphs form an interesting subclass of series--parallel graphs but such a class is not in general closed under taking the dual. A \emph{self--dual outer-planar graph} is a 2--connected outer--planar graph whose plane dual is also outer-planar.\ In this section we characterize self--dual outer--planar graphs and their fundamental graphs. Interestingly, such a characterization implies a characterization for bipartite chain graphs. The heart of the characterization is Theorem~\ref{thm:arrowt2closed} establishing that $\{arrow,T_2\}$--free BDH graphs form a class of bipartite graphs closed under pivoting. To prove the characterization we need some intermediate results.\ Besides their own interest such result allows us to spare a tedious case analysis in the proof of Theorem~\ref{thm:arrowt2closed} and at the same time give some light on \emph{bisimplicial edges} in BDH graphs.
\mybreak
Recall first the notion of \emph{bisimplicial edge} introduced by Golumbic and Goss in \cite{golgoss} as a bipartite analogue of the very well known simplicial vertices in graphs.
\begin{definition}
	An edge $uv$ of a bipartite graph $G$ is a \emph{bisimplicial} edge of $G$ if the subgraph induced in $G$ by $N_G(u)\cup N_G(v)$ is a complete bipartite graph.
\end{definition}  
A \emph{proper bisimplicial} edge of $G$ is a bisimplicial edge that is not a pending edge.\ While, as shown by Golumbic and Goss~\cite{golgoss}, bipartite chordal graphs and, in particular BDH, always have bisimplicial edges, there is in general no warranty that they have nontrivial bisimplicial edges.
\begin{lemma}\label{lem:arrowfree}
	Let $G$ be a bipartite chordal graph. If $G$ is arrow-free, then either $G$ is a tree or it contains a proper bisimplicial edge. 
\end{lemma}
\begin{proof} Let $G$ have color classes $A$ and $B$. Since $G$ is chordal, if $G$ is not a tree, then it contains a set of vertices that induces a square. Hence $G$ contains an induced complete bipartite graph $G^*$ whose color classes $A^*$ and $B^*$ have at least two vertices each and are such that $A^*\cup B^*$ is inclusion-wise maximal. We claim that 
	\mybreak
	\textit{Claim} there are $u^*\in A^*$ and $v^*\in B^*$ such that $N_G(u^*)=A^*$  and $N_G(v^*)=B^*$.
	\mybreak
	\textit{Proof of the claim}. If $u'$ and $u''$ are distinct vertices in $A^*$, then either $N_G(u')\subseteq N_G(u'')$ or $N_G(u'')\subseteq N_G(u')$. To see this observe that if there existed $x\in N_G(u'')\setminus N_G(u')$ and $y\in N_G(u')\setminus N_G(u')$ then for any two distinct vertices $v'$ and $v''$ in $B^*$ the set $\{u',u'',v',v'',x,y\}$ would induce an arrow in $G$. By the same reason, if $v'$ and $v''$ are distinct vertices in $B^*$, then either $N_G(v')\subseteq N_G(v'')$ or $N_G(v'')\subseteq N_G(v')$.  We therefore conclude that the vertices in $A^*$ have pairwise inclusion-wise comparable neighborhood as well as the vertices in $B^*$. Hence, for some $u^*\in A^*$ one has $N_G(u^*)=\bigcap_{u\in A^*}N_G(u)$ and for some $v^*\in B^*$ one has $N_G(v^*)=\bigcap_{v\in B^*}N_G(v)$. Clearly $A^*\subseteq \bigcap_{v\in B^*}N_G(v)$ and $B^*\subseteq \bigcap_{u\in A^*}N_G(u)$ because $G'$ is complete. On the other hand both 
	$$A^*\cup \left(\bigcap_{u\in A^*}N_G(u)\right)$$
	and
	$$\left(\bigcap_{v\in B^*}N_G(v)\right)\cup B^*$$
	induce complete bipartite subgraphs of $G$. Therefore $A^*=\bigcap_{v\in B^*}N_G(v)$ and $B^*=\bigcap_{u\in A^*}N_G(u)$
	by the maximality of $A'\cup B'$. Consequently $N_G(u^*)=A^*$  and $N_G(v^*)=B^*$, as claimed.
	\hspace*{\fill}$\blacksquare$
	\mybreak
	By the claim if $G$ is not a tree, then there are two vertices $u^*\in A$ and $v^*\in B$ such that $N_G(u^*)\cup N_G(v^*)$ induces a complete bipartite subgraph of $G$ whose color classes have at least two vertices each.  Hence $u^*v^*$ is a proper bisimplicial edge as required.
\end{proof} 
Let us understand now the behavior of bisimplicial edges under pivoting. 
\begin{lemma}\label{lem:bisim}
Let $G$ be BDH graph and let $e,\,f\in E(G)$,  $e=xy$. 
\begin{enumerate}[(a)]
	\item\label{com:bisia} If $e$ is a bisimplicial edge of $G$, then $e$ is a cut--edge of $G^e$ and, conversely, if $e$ is a cut--edge of $G$, then $e$ is a bisimplicial edge of $G^e$. 
	\item\label{com:bisib} If $f$ is not adjacent to $e$ and $f$ is not a proper bisimplicial edge of $G$, then $f$ is not a proper bisimplicial edge of $G^e$.
	\item\label{com:bisic} If $G-\{x,y\}$ is connected and $p$, $u$ and $v$ are three vertices of $G$ such that $p$ and $u$ are leaves of the same color in $G-\{x,y\}$ and $v$ is adjacent in $G-\{x,y\}$ to $p$ but not to $u$, then $\{u,v\}$ does not induce a bisimplicial edge in $G^e$.
\end{enumerate}
\end{lemma}
\begin{proof}~~
	\begin{itemize}
		\item[\eqref{com:bisia}] If $e$ is a pending edge the statement is trivial because $G\cong G^e$ and $G^e-e$ separates the end-vertex of degree 1 of $e$ from the rest of the graph. We can therefore assume that $e$ is a proper bisimplicial edge. Let $A^\ast=N_G(y)$ and $B^\ast=N_G(x)$ and observe that both $A^\ast$ and $B^\ast$ have at least two vertices each because $e$ is nontrivial. Since $e$ is bisimplicial, the subgraph induced in $G^e$ by $(A^\ast\setminus u)\cup B^\ast\setminus y$ is edgeless. Suppose to the contrary that $e$ is not a cut edge of $G^e$. Hence $G^e-e$ is connected and $\textrm{dist}_{G^e-e}$ is finite over $V(G^e)$. Since no edge of $G^e$ connects vertices of $A^\ast\setminus x$ to vertices in $B^\ast\setminus y$ it follows that $G^e-\{x,y\}$ is a connected induced subgraph of $G^e$ and 
		$$\textrm{dist}_{G^e-e}(u,v)=\textrm{dist}_{G^e-\{x,y\}}(u,v)$$
		for each $u\in A^\ast\setminus x$ and each $v\in B^\ast\setminus y$.
		On the other hand, by Corollary~\ref{cor:main}, $G^e$ is BDH hence, by definition, $\textrm{dist}_{G^e}$ restricts over the induced connected subgraphs of $G^e$ and since $G^e-\{x,y\}$ is such, it follows that 
		$$\textrm{dist}_{G^e-\{x,y\}}(u,v)=\textrm{dist}_{G^e}(u,v)=3$$
		for each $u\in A^\ast\setminus x$ and each $v\in B^\ast\setminus y$. Therefore we conclude that if $x_0$ and $y_0$ are arbitrarily chosen vertices in $A^\ast\setminus x$ and in $B^\ast\setminus y$, respectively, then in $G^e$ there is a chordless path $P\cong P_4$ joining $x_0$ and $y_0$ and whose inner vertices, $s$ and $t$, say, do not belong to $A^\ast\cup B^\ast$. But now $\{x,y,x_0,y_0,s,t\}$ induces a chordless cycle $C\cong C_6$ in $G^e$ contradicting that $G^e$ is BDH. We conclude that $G^e-e$ is not connected and hence that $e$ is a cut edge of $G^e$.
		\item[\eqref{com:bisib}] The statement is true when $e=f$ because if $e$ is not a proper bisimplicial edge, then $e$ is nonpending and non bisimplicial Therefore $N(x)\cup N(y)$ does not induce a complete bipartite subgraph and so neither does its bipartite complement. The statement is also trivially true when $e\not=f$ and $e$ is a pending edge because in this case $G\cong G^e$ and all edges of $G$ are preserved by pivoting. We can therefore assume that $e\not=f$ and $e$ is nonpending.\ Let $f=uv$, with $x$ and $u$ in the same color class, say. Since pivoting complements the edges between $N_G(u)$ and $N_G(v)$, if $f$ were proper bisimplicial in $G^e$, then $N_G(x)\cap N_G(u)\not=\emptyset$ and $N_G(y)\cap N_G(v)\not=\emptyset$. Since $f\in E(G^e)$, necessarily $\{u,v\}\not\subseteq N_G(x)\cup N_G(y)$ otherwise $f$ would be complemented in $G^e$.\ Hence at most one of $xv$ and $yu$ can be in $G$. Let $w\in N_G(x)\cap N_G(u)$ and $z\in N_G(y)\cap N_G(w)$. We first note that if $wz$ is in $G$ then $wz$ is not in $G^e$, hence $f$ cannot be simplicial in $G^e$. So, let us assume that $wz$ is not in $G$: if none of $xv$ and $yu$ is in $E(G)$, then $\{x,y,u,v,z,w\}$ induces a hole on six vertices in $G$ otherwise, namely if one of $xv$ and $yu$ is in $G$, the same set of vertices induces a domino. 
		\item[\eqref{com:bisic}] Suppose to the contrary that $uv$ is a bisimplicial edge of $G^e$. Since $uv$ is not an edge of $G$, then $\{u,v\}\subseteq N_G(x)\cup N_G(y)$. Hence $\dist_G(u,v)=3$. By the connectedness of $G-\{x,y\}$ and because such a graph is BDH, it follows that $\dist_{G-\{x,y\}}(u,v)=3$ and there is a chordless path $P$ of length 3 connecting $u$ and $v$ such that $x$, $y$ and $e$ are not traversed by $P$.\ Therefore $\{x,y\}\cup V(P)$ is the set of vertices of a (not necessarily induced) cycle $C$ on six vertices. Again, since $G$ is BDH, the vertices of $C$ induces in $C$ at least two chords and hence exactly two chords because $u$ and $v$ are not adjacent. Observe that the subgraph induced by $V(C)$ in $G$ is isomorphic to $K_{3,3}-\bar{e}$, namely $K_{3,3}$ with one edge missing, where the missing edge is $uv$. Let $N$ be the subgraph of $G$ induced by $V(C)\cup\{p\}$. By the hypothesis, $p$ and $u$ are leaves of the same color in $G-\{x,y\}$. In particular $p$ has degree at most 2 in $G$ and, if it has degree 2, then $p$ is adjacent to $y$ in $G$. Suppose first that $p$ had degree 1 in $G$. In this case $uv$ would not be a bisimplicial edge in $G^e$, because $p\in N_{G^e}(v)$, $y\in N_{G^e}(u)$ while $py\not\in E(G^e)$. Hence the only possibility left is that $py\in E(G)$ and hence $py\in E(G^e)$ because the edges of $G$ adjacent to the pivot edge $e$ are retained in $G^e$. In this case $\{y,u,q_2,q_3,v,p\}$ would induce a domino contradicting that $G$ is domino-free. We conclude that $\{u,v\}$ does not induce a bisimplicial edge of $G^e$. 
		\end{itemize}
\end{proof}
Let $G$ a bipartite graph and $N$ be an induced subgraph of $G$. Let $e\in E(G)$ have endpoints $x$ and $y$. By definition of pivoting, $N$ is affected by pivoting on $e$ if the neighborhood of $x$ and $y$ intersect $V(N)$. More precisely, denote by $N_e$ the subgraph of $G$ induced by $V(N)\cup\{x,y\}$. Then 
\cadre\label{cadre:1}
if one of $N_{G}(x)\cap V(N)$ and $N_{G}(y)\cap V(N)$ is empty, then the subgraph induced by $V(N)$ in $N_e^e$ (and hence in $G^e$) is isomorphic to $N$.
\endcadre
\mybreak
The following result is needed in the proof of Theorem~\ref{thm:arrowt2closed}.
\begin{lemma}\label{lem:paoloinduced}
Let $N_e$ be an induced subgraph of a BDH graph $G$, where $e=xy$ and $N$ is connected. Suppose that edge $f=xu$, $u\not=y$, is a proper bisimplicial edge of $N_e$. Then $u$ and $y$ are \emph{twins} in $N_e$, namely, they have the same neighbors in $N_e$.
\begin{proof}
By the bisimpliciality of $f$ it follows that each neighbor of $u$ in $N_e$ is a neighbor of $y$. Suppose that $u$ and $y$ are not twins. Hence there is a vertex $z\in V(N)$ adjacent to $y$ but not to $u$. Since $u\in V(N)$ and $N$ is connected, there is a chordless path $P$ joining $z$ and $u$. Let $w$ be the unique neighbor of $u$ on $P$. Then $w$ is adjacent to $y$. Let $t\not=u$ be the neighbor of $y$ closest to $w$ (possibly $t=z$) in $P$ and let $Q$ be the subpath of $P$ joining $t$ and $w$. If $x$ has no neighbor on $Q$, then $V(Q)\cup\{u,x,y\}$ induces a subgraph of $N_e$ which is not BDH. Otherwise let $s\not=u$ be the neighbor of $x$ closest to $w$ on $Q$ and let $R$ be the subpath of $Q$ joining $s$ to $w$.Then $V(R)\cup\{u,x,y\}$ induces a subgraph of $N_e$ which is not BDH. In both case the fact that $N_e$ is BDH would be contradicted we therefore conclude that $u$ and $y$ are twins.
\end{proof}
\end{lemma}
\begin{theorem}\label{thm:arrowt2closed}
If $G$ is $\{hole,domino,arrow,T_2\}$--free graph, then so is any graph in the orbit of $G$. Equivalently, the class of $\{arrow,T_2\}$--free BDH graphs is closed under pivoting.
\end{theorem}
\begin{proof}
It suffices to show that if $G$ is an $\{arrow,T_2\}$--free BDH graph, then so is the graph $G^e$ for each $e\in E(G)$. Equivalently, using the fact that $(G^e)^e\cong G$ for each $e\in E(G)$, it suffices to prove that if $G$ is BDH and contains an induced subgraph $N$ which is either isomorphic to an arrow or to $T_2$, then $G^e$ is not an $\{arrow,T_2\}$--free BDH for each $e\in E(G)$.
\mybreak
We prove the theorem in the latter formulation by showing that 
\cadre\label{cadre:2}
if $N$ is either isomorphic to an arrow or to $T_2$, then for each $e\in E(G)$, $N_{e}^e$ either is not a BDH graph or it contains an induced copy $L$ of an arrow or of $T_2$. 
\endcadre\mybreak
Indeed, by Lemma~\ref{lem:0}, $N_e^e$ is an induced subgraph of $G^e$ and hence $L$ is an induced subgraph of $G^e$ being an induced  subgraph of $N_e^e$.
\mybreak
By \eqref{cadre:1}, \eqref{cadre:2} is trivially true when at least one of the neighborhoods of the endpoints of $e$ has empty intersection with $V(N)$. So, if $e=xy$, we may assume that both $N_G(x)\cap V(N)$ and $N_G(y)\cap V(N)$ are nonempty. Accordingly we distinguish three cases.
\mybreak
\underline{\textbf{Case 1}: $e\in E(N)$}. This case is easily ruled out by the following claim.

\begin{claim}\label{claim:a} If $N$ is an arrow, then, for $e\in E(N)$, $N^e$ is still an arrow and if $N\cong T_2$ then, for $e\in E(N)$, either $N^e\cong T_2$ or $N^e$ contains an induced arrow. 
\end{claim}
\textit{Proof of Claim \ref{claim:a}}. It is just a matter of checking and, by symmetry, it suffices to check the claim only for one edge of the square, if $N$ is an arrow, or for one of the edges incident to the center of $N$, if $N\cong T_2$--note that if $e$ is a pending edge, then the claim is trivial. 
\hspace*{\fill}$\blacksquare$
\mybreak
\underline{\textbf{Case 2}: $e$ is incident to a vertex of $N$}. Without loss of generality let $x\in V(N)$ and $y\not\in V(N)$. This case is more complicated than the previous one.
Suppose first that $N$ is a tree. Since $N_e$ has to be distance hereditary and since $N$ is an induced subgraph of $N_e$, because $N=N_e-y$, it follows that $y$ can be adjacent to only one vertex of $N$ and such a vertex is at distance 2 from $x$. Let $v$ be the unique neighbor of $y$ in $N$ and let $z$ be the central vertex of the unique path connecting $x$ to $v$. There are only two possibilities: either one among $x$ and $v$ is a leaf of $N$ or neither is. In the former case $N_e-z\cong T_2$ and $N_e-z$ is an induced subgraph of $N_e$ with $e\in E(N_e)$.\ By Lemma~\ref{lem:0}, $(N_e-z)^e$ is an induced subgraph of $N_e^e$. In the latter case, namely when neither $x$ nor $v$ are leaves, $x$ has degree at least two and so does $v$. Hence there are nonadjacent vertices $x'$ and $v'$ adjacent to $x$ and $v$, respectively, such that $\{x',x,z,v,v',y\}$ induces an arrow $L$ in $N_e$ with $e\in E(L)$.\ Therefore, by Claim~\ref{claim:a} and still by Lemma~\ref{lem:0}, $L^e$ is an induced arrow of $N^e_e$. We conclude that in both cases $N_e^e$ is not $\{arrow,T_2\}$--free.
\mybreak
Suppose now that $N$ is an arrow.\ Suppose that the vertices of $N$ are labeled as follows: the pending vertices of $N$ by $p$ and $p'$; the vertices of degree 3 by $v$ and $v'$ with $p$ adjacent to $v$ and $p'$ adjacent to $v'$, while those of degree 2 by $z$ and $z'$.\ By symmetry we have to consider only three possibilities.
\mybreak
--\emph{$e$ is incident to a pending vertex}.\ Without loss of generality $x=p$. Now $y$ is adjacent to at least one among $z$, $z'$ and $p'$.\ It easily checked however that if $z$ and $z'$ are not both neighbors of $y$, then $N_e$ is not a BDH graph because it contains induced $C_6$ or dominoes. The same happens if $y$ is adjacent to $p'$, no matter whether $y$ has other neighbors. The unique case which does not contradict that $N_e$ is BDH occurs when the neighborhood of $y$ consists precisely of $z$ and $z'$. In this case however $N_e-v$ is an induced arrow of $N_e$ with $e\in E(N_e-v)$.\ By Lemma~\ref{lem:0}, $(N_e-v)^e$ is an induced subgraph of $N_e^e$ and we conclude that $N_e^e$ is not $\{arrow,T_2\}$--free.

\noindent --\emph{$e$ is incident to a vertex of degree 2}.\ Without loss of generality $x=z$. In this case, the set of neighbors of $y$ is either $\{z'\}$ or $\{s,z'\}$ where $s$ is either $p$ or $p'$.\ In any other case, indeed, $N_e$ contains an induced domino having $e$ among its edges.\ Without loss of generality, $s=p$. If $y$ is adjacent to $z'$ only, then $N_e$ is isomorphic to $K_{2,3}$ with two pending edges appended to the vertices of degree 3. It is therefore easily checked that $N_e^e\cong T_2$. If the neighbors of $y$ are $p$ and $z'$, then $N_e-v$ is an induced arrow of $N_e$ with $e\in E(N_e-v)$.\ By Lemma~\ref{lem:0}, $(N_e-v)^e$ is an induced subgraph of $N_e^e$ and we conclude that $N_e^e$ is not $\{arrow,T_2\}$--free.

\noindent --\emph{$e$ is incident to a vertex of degree 3}.\ Without loss of generality $x=v$. In this case, $y$ must be adjacent to $v'$. Now $z$ has degree 2 in $N_e$ as well and $N_e-z$ in an induced arrow of $N_e$ with $e\in E(N_e-z)$. By Lemma~\ref{lem:0}, $(N_e-z)^e$ is an induced subgraph of $N_e^e$ and we conclude that $N_e^e$ is not $\{arrow,T_2\}$--free.
\mybreak
So \eqref{cadre:2} is true in any of the subcases above and therefore in case 2.
\mybreak
\underline{\textbf{Case 3}: $e$ is disjoint from $N$}. In this case $e$ can interact with $N$ in several (non-isomorphic) ways---we counted at least 15 possibilities.\ Therefore, rather than pursuing a case analysis, we use our previous lemmas. Suppose to the contrary that $N_e^e$ is an $\{arrow,T_2\}$--free BDH graph. By Lemma~\ref{lem:arrowfree}, $N_e^e$ is either a tree or possesses a proper bisimplicial edge. By Lemma \ref{lem:0} one has
$$N_e^e-\{x,y\}\cong\left(N_e-\{x,y\}\right)^e\cong N^e.$$
Claim \ref{claim:a} implies that $N^e$is connected.\ Therefore  $N_e^e$ is not a tree, because $e$ is not a pending edge. Hence $N_e^e$ possesses a proper bisimplicial edge by Lemma~\ref{lem:arrowfree}.
\mybreak
Observe first that $e$ is not a proper bisimplicial edge of $N_e^e$ because, by Part~\ref{com:bisia} of Lemma~\ref{lem:bisim}, if $e$ were proper bisimplicial in $N_e^e$, then $e$ would be a cut edge of $N_e=(N_e^e)^e$; but $e$ cannot be a cut edge of $N_e$ because $N$ is connected and both the endpoints of $e$ have neighbors in $N$. Therefore the proper bisimplicial edges of $N_e^e$ are either of the following types:
\begin{itemize}
\item[--] proper bisimplicial edges of $N$ retained in $N^e$, that is, preserved by pivoting;
\item[--] proper bisimplicial edges of $N^e$ connecting non adjacent vertices of $N$, that is, created by pivoting;
\item[--] proper bisimplicial edges of $N_e^e$ adjacent to $e$ in $N_e^e$.
\end{itemize}
The proper bisimplicial edges of $N_e^e$ cannot be of the first type. To see this observe that since $N$ is an arrow or $T_2$, it does not have proper bisimplicial edges. In particular the edges of $N$ are not proper bisimplicial edges of $N_e$ and by Part~\ref{com:bisib} of Lemma~\ref{lem:bisim}, they cannot become proper bisimplicial in $N_e^e$ by pivoting. 

The proper bisimplicial edges of $N_e^e$ are not even of the second type. To see this observe that, since $N_e-\{x,y\}=N$ and $N$ is either an arrow or $T_2$,  Part~\ref{com:bisic} of Lemma~\ref{lem:bisim} applies and we deduce that no pair of nonadjacent vertices of $N$ is connected by a proper bisimplicial edge of $N_e^e$. 

We therefore conclude that the bisimplicial edges of $N_e^e$ can be only of the third type, namely, those edges adjacent to $e$ in $N_e^e$ (and in $N_e$ as well). We show that not even this case can occur thus contradicting Lemma~\ref{lem:arrowfree} and at the same time completing the proof. To this end let $f$ be an edge of $N_e^e$ adjacent to $e$. Say, $f=uv$ with $x=v$. Suppose also that $f$ is a proper bisimplicial edge of $N_e^e$. As observed above, $N_e^e-\{x,y\}$ is connected and therefore Lemma~\ref{lem:paoloinduced} applies (with $N_e^e-\{x,y\}$ in place of $N$ and $N_e^e$ in place of $N_e$). By the lemma, $y$ and $u$ are twins. Therefore
$$N_e^e-u\cong N_e^e-y\cong \left(N_e-y\right)^e$$ 
where the rightmost isomorphism is due to Lemma \ref{lem:0}. Since $N^e$ is an induced subgraph of $\left(N_e-y\right)^e$, it follows that $N_e^e-u$ contains an induced copy of $N^e$. By Claim \ref{claim:a}, $N^e$ is either an arrow or $T_2$ or contains an induced arrow. Hence we conclude that $N_e^e-u$ (and hence $N_e^e$) is not $\{arrow,T_2\}$--free as we are assuming. This contradiction completes the proof of the theorem. 
\end{proof}
To complete the characterization we need only a last device about the weak dual of self--dual outerplanar graphs.
\begin{lemma}\label{lemma:selfdual}
A graph is a self--dual outer-planar graph if and only if its weak-dual is a path.
\end{lemma}
\begin{proof}
Clearly if the weak-dual of a 2--connected planar graph $H$ is a path, then the dual $H^*$ of $H$ is outer-planar and $H$ is outer-planar as well (see Figure~\ref{fi:weakd}(i))---notice that the endvertices of the path are adjacent to the vertex corresponding to the outer face by 2--connectedness. Hence $H$ is a self--dual outer-planar graph. Conversely,  
if $H$ is a self--dual outer-planar graph, then the weak dual $\overline{H^*}$ of $H$ must be a path. To see this let $w$ be the vertex of $H^*$ corresponding to the outer face of an out-plane embedding of $H$. We know that $\overline{H^*}$ is a tree. Suppose that $\overline{H^*}$ has a vertex $u$ with at least three neighbors $v_1$, $v_2$ and $v_3$. Then $\overline{H^*}-u$ has at least three connected components $B_1$, $B_2$ $B_3$, say, with $v_1\in B_1$, $v_2\in B_2$ and $v_3\in B_3$. Since $H$ is 2--connected so is $H^*$. Hence, by 2--connectedness, there are edges $wb_1$, $wb_2$ and $wb_3$, where $b_i\in B_i$, for $i=1,2,3$. We can choose the $b_i$'s so that among the neighbors of $w$, $b_i$ is the closest to $v_i$ in $B_i$, $i=1,2,3$. Let $P_i$ be the unique path in $\overline{H^*}$ connecting $v_i$ to $b_i$. Then, the subgraph of $H^*$ induced by $V(P_1)\cup V(P_2)\cup V(P_3)\cup\{u,w\}$ is a subdivision of $K_{2,3}$.
\end{proof}

\begin{figure}[H]
    \begin{center}
             \includegraphics[width=13cm]{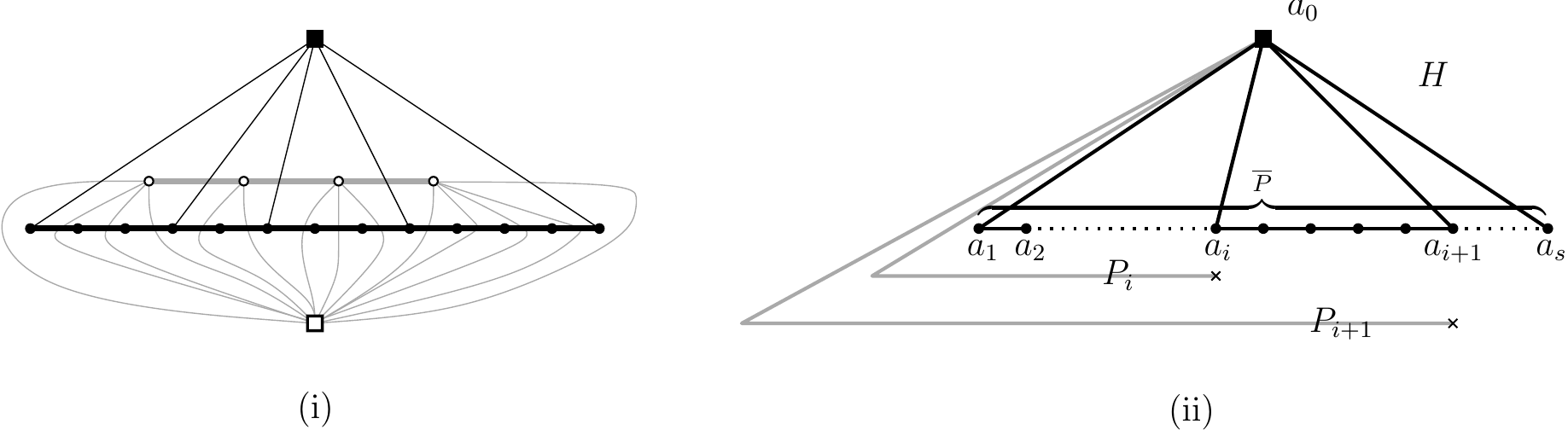}
    \end{center}
    \caption{(i) A self--dual outer-planar graph (black) and its dual (gray); the respective weak duals are depicted thickly; (ii) $P_i$ and $P_{i+1}$ are two nested fundamental paths with respect to the tree $\overline{P}\cup\{a_0a_1\}$ in the outer-planar graph $H$.}
    \protect\label{fi:weakd}
\noindent\hrulefill%
\end{figure}
\begin{theorem}\label{thm:selfdualouterplanar}
Let $G$ be a bipartite graph with color classes $A$ and $B$. Then the following statements are equivalent.
\begin{enumerate}
\item $G$ is $\{hole,domino,arrow,T_2\}$--free.
\item $G$ is the fundamental graph of a self--dual outerplanar graph.
\item $G$ is pivot-equivalent to a chain graph.
\end{enumerate}
\end{theorem}
\begin{proof}~~\\
1.$\Rightarrow$2. If $G$ is hole and domino-free, then it is the fundamental graph of a 2--connected series--parallel graph by Theorem~\ref{thm:main}. Hence there is a 2--connected series--parallel graph $H$ such that $G\cong B_T(H)$ where $T$ is a spanning tree of $H$. At the same time $G\cong B_{T^*}(H^*)$ where $H^*$ is a dual of $H$ and $T^*$ is a co-tree of $H$. $H^*$ is a series--parallel graph. Observe now that arrows are fundamental graphs of graphs isomorphic to $K_{2,3}$. Therefore, neither $H$ nor $H^*$ can contain minor isomorphic to $K_{2,3}$ because the fundamental graphs of these minor would contain as induced subgraphs graphs in the orbit of $T_2$ or arrows (by Lemma~\ref{lem:matroidorbit}), contradicting that $G$ is $\{Hole,Domino,Arrow,T_2\}$--free.\ It follows that both $H$ and $H^*$ are outerplanar graphs and therefore $G$ is the fundamental graph of a self--dual outer-planar graph. 
\mybreak
2.$\Rightarrow$3. By definition there are two 2--connected outer-planar graphs $H$ and $H^*$ such that $G\cong B_T(H)$ for some spanning tree $T$ of $H$ and $G\cong B_{T^*}(H^*)$ where $H^*$ is a dual of $H$ and $T^*$ is a co-tree of $H$. Let $\overline{P}^*$ be a weak dual of $H$ and $\overline{P}$ be a weak dual of $H^*$. According to Lemma~\ref{lemma:selfdual}, both $\overline{P}$ and $\overline{P}^*$ are paths. Let $\overline{P}=(\{a_1,\ldots,a_s\},\{a_1a_2,\ldots,a_{s-1}a_s\})$ and $\overline{P}^*=(\{b_1,\ldots,b_t\},\{b_1b_2,\ldots,b_{t-1}b_t\})$ (Figure~\ref{fi:weakd}(ii)). Let $a_0$ be the vertex of $H$ corresponding to the outer face of $H^*$ and let $b_0$ be the vertex of $H^*$ corresponding to the outer face of $H$. 

\noindent There are edges $a_0a_1$ and $a_0a_s$ in $H$ and edges $b_0b_1$ and $b_0b_t$ in $H^*$ by 2--connectedness. The cycle $C^*$ spanned by $E(\overline{P})\cup\{a_0a_1,a_0a_s\}$ is the unique Hamiltonian cycle of $H$ and the cycle $C^*$ spanned by $E(\overline{P}^*)\cup\{b_0b_1,b_0b_t\}$ is the unique Hamiltonian cycle of $H^*$. Hence $E(\overline{P})\cup\{a_0a_1\}$ and $E(\overline{P}^*)\cup\{b_0b_1\}$ span Hamiltonian paths in $H$ and $H^*$ respectively, and such paths are spanning trees of $H^*$ and $H$, respectively. Denote them by $P$ and $P^*$, respectively. Therefore, up to  re-labellings, $B_P(H)$ and $B_{P^*}(H^*)$ are isomorphic graphs and thus isomorphic (as unlabeled graphs) to the same graph $\Gamma$ and $G\sim \Gamma$.\ It remains to show that $\Gamma$ is a chain graph. $\Gamma$ have color classes $A$ and $B$. There is a bijection $\psi: A\cup B\rightarrow E(H)$ such that the image of $\{u\}\cup N_\Gamma(u)$ is a fundamental circuit in $H^*$ with respect to $P^*$ for each $u$ in $A$ and the image of $\{v\}\cup N_\Gamma(v)$ is a fundamental circuit in $H$ with respect to $P$ for each $v$ in $B$. Hence, for each $u\in A$, the image of $N_G(u)$ is a subpath of $P^*$ of the form $(\{b_0,b_1,\ldots,b_{l(u)}\},\{b_0b_1,\ldots,b_{l(u)-1}b_{l(u)}\})$ for some $1\leq l(u)\leq t$ and for each $v\in B$, the image of $N_G(v)$ is a subpath of $P$ of the form $(\{a_0,a_1,\ldots,a_{m(u)}\},\{a_0a_1,\ldots,a_{m(u)-1}b_{m(u)}\})$ for some $1\leq m(u)\leq s$. Let $P_A=\psi(P)$ and $P_B=\psi(P^*)$. We conclude that both families $(N_G(u), u\in A)$ and $(N_G(v), v\in B)$ are linearly ordered by inclusion because they consist of edge--sets of a of nested subpaths of $P_A$ and $P_B$, respectively. Therefore $\Gamma$ is a chain graph in the orbit of $G$.
\mybreak
3.$\Rightarrow$1.
Since $G$ is pivot-equivalent to a chain graph, it follows that there exists a graph $\widetilde{G}$ in the orbit of $G$ which is $2K_2$--free. Thus $\widetilde{G}$ cannot contain induced copies of holes, dominoes, arrows and $T_2$ as the latter graphs all contain induced copies of $2K_2$. By Theorem~\ref{thm:arrowt2closed} each graph in the orbit of $\widetilde{G}$ is $\{Hole,Domino,Arrow,T_2\}$--free.\ In particular $G$ is such. 
\end{proof}

\subsection{Packing paths and multi commodity flows in series--parallel graphs.}\label{sec:flow}
In this section we give an application of Theorem~\ref{thm:main} in Combinatorial Optimization. We show that a notoriously hard problem contains polynomially solvable instances when restricted to series--parallel graph. Let $H=(V,E)$ be a graph and let $F\subseteq E$ be a set of prescribed edges of $H$ called the \emph{nets} of $H$. Following \cite{FKAndras} a path $P$ of $H$ will be called
\emph{$F$-admissible} if it connects two vertices $s,t$ of $V$ with $st\in
F$ and $E(P)\subseteq E-F$. Let $U$ be the set of end-vertices of the nets. In the context of network-flow vertices of $U$ are thought of
as terminals to be connected by a flow of some commodity (the nets are in fact also known as \emph{commodities}). Let $\mathcal{P}_F$ denote
the family of all $F$-admissible paths of $G$ and let
$\mathcal{P}_{F,f}\subseteq \mathcal{P}_F$ be the family of those
$F$-admissible paths connecting the endpoints $s$,$t$ of net $f$. An $F$-\emph{multiflow} (see e.g. \cite{bible}), is a
function $\lambda: \mathcal{P}_F\rightarrow \mathbb{R}_+$, $P\mapsto \lambda_P$. The
multiflow is integer if $\lambda$ is integer valued. The value of
the $F$-\emph{multiflow on the net $f$} is $\phi_f=\sum_{P\in
\mathcal{P}_{F,f}}\lambda_P$. The total value of the multiflow is
the number $\phi=\sum_{f\in F}\phi_f$. Let $w:E-F\rightarrow
\mathbb{Z}_+$ be a function to be though of as a capacity
function. An $F$-\emph{multiflow subject to $w$} in $H$ is an $F$-multiflow  such
that,
\begin{equation}
\sum_{P\in \mathcal{P}_F: E(P)\ni e}\lambda_P\leq w(e), ~~ \forall e\in E-F
\label{eq:flowcapacity}
\end{equation}
When $w(e)=1$ for all $e\in E-F$, an integer multiflow is simply a
collection of edge--disjoint $F$-admissible paths of $H$. The $F$-\emph{Max-
Multiflow Problem} is the problem of finding, for a given capacity
function $w$, an $F$-multiflow subject to $w$ of maximum total value.
An $F$-\emph{multicut} of $H$ is a subset of $B$ edges of $E-F$ that intersects the edge--set of each $F$-admissible path.\ The name $F$-multicut is due to the fact that the removal of the edges of $B$ from $H$ leaves a graph with no $F$-admissible path that is, in the graph $H-B$ it is not possible to connect the terminals of any net. The 
\emph{capacity of the $F$-multicut} $B$ is the number $\sum_{e\in B}w(e)$. 
\mybreak
Multiflow Problems are very difficult problems (see
\cite{franksurvey}, \cite{FKAndras} and Vol. C, Chapter 70 in
\cite{bible}). In \cite{gargvazi} it has been shown that the
Max-Multiflow Problem is NP-hard even for trees and even for
$\{1,2\}$-valued capacity functions. The problem though, is shown
to be polynomial time solvable for constant capacity functions by
a dynamic programming approach. However, even for constant
functions, the linear programming problem of maximizing the value of the multiflow
over the system of linear inequalities (\ref{eq:flowcapacity}) has
not even, in general, half-integral optimal solutions. In \cite{NVZ}, the NP-completeness of the
Edge--Disjoint--Multi commodity Path Problem for series--parallel
graph (and partial 2--trees) has been established while, previously
in \cite{ZTN}, the polynomial time solvability of the same problem
for partial 2--trees was proved under some restriction either on
the number of the commodities (required to be a logarithmic
function of the order of the graph) or on the location of the
nets. Using our results we give a further contribution to the above
mentioned problems.
\begin{theorem} Let $H=(V,E)$ be a 2--connected series--parallel graph and let $F$ be the edge--set of any of its spanning co-trees. Then the maximum total value of an $F$-multiflow equals the minimum capacity of an $F$-multicut. Furthermore, both a maximizing multiflow and a minimizing multicut can be found in strongly
polynomial time.
\end{theorem}
\begin{proof} 
Let $\mv{A}$ be a $\{0,1\}^{m\times n}-$valued matrix and $\mv{b}\in \mathbb{Z}_+^m$ be a vector. Let $\mv{1}_d$ be the all ones vector in $\mathbb{R}^d$. Consider the linear programming problem 
\begin{equation}\label{eq:lp1}
\max_{\mv{x}\in \mathbb{R}_+^n} \left\{ \mv{1}_n^T\mv{x} \mid \mv{A}\mv{x}\leq \mv{b} \right\}\end{equation}
and its dual 
\begin{equation}\label{eq:lp2}
\min_{\mv{y}\in \mathbb{R}_+^m}\left\{ \mv{b}^T\mv{y} \mid \mv{A}^T\mv{y} \geq \mv{1}_n\right\}.
\end{equation}
By the results of Hoffman, Kolen and Sakarovitch~\cite{hokosa} and Farber~\cite{farber}, if $\mv{A}$ is a totally balanced matrix (i.e., $\mv{A}$ is the reduced adjacency matrix of a bipartite chordal graph), then 
both the linear programming problem above have integral optimal solutions and, by linear programming duality, the two problems have the same optimum value. Furthermore, an integral optimal solution $\mv{x}^*$ to the maximization problem in \eqref{eq:lp1} satisfying the additional constraint 
\begin{equation}\label{eq:simpleflow}
\mv{x}^*\leq \mv{1}_n
\end{equation}
and an integral optimal solution $\mv{y}^*$ to the minimization problem in \eqref{eq:lp2} satisfying the additional constraint 
\begin{equation}\label{eq:simplecut}
\mv{y}^*\leq \mv{1}_n
\end{equation} can be found in strongly polynomial time. 
\mybreak
Let now $H$ be a 2--connected graph and let $F$ be the edge--set of a co--tree $\overline{T}$ of some spanning $T$ tree of $H$. By giving a total order on the edge--set of $T$, one can define a vector $\mv{b}$ whose entries are the values of the capacity function $w:E(H)-F\rightarrow \mathbb{Z}_+$. If $\mv{A}$ is the the incidence matrix of  $\mathcal{P}_F$, namely the matrix whose columns are the incidence vectors of the $F$-admissible paths of $H$, then $\mv{A}$ is a partial representation of $M(H)$. Moreover, if $H$ is series--parallel, then $\mv{A}$ is totally balanced: by Theorem~\ref{thm:main}, $\mv{A}$ is the reduced adjacency matrix of a BDH graph which is chordal being hole-free (by Theorem~\ref{thm:bdhchar}). On the other hand, integral solutions to the
problem in \eqref{eq:lp1} satisfying constraint~\eqref{eq:simpleflow} and to the problem in \eqref{eq:lp2} satisfying constraint~\eqref{eq:simplecut} are incidence vectors of $F$-multiflows and $F$-multicuts, respectively. Hence, both an $F$-multiflow of maximum value and an $F$-multicut of minimum capacity can be found in strongly polynomial-time by solving the linear programming problems above. Moreover, linear programming duality implies that the maximum value of an $F$-multiflow and the minimum capacity of an $F$-multicut coincide.
\end{proof}

\section{One more proof of Theorem~\ref{thm:main}}
\mybreak
In this section, we give another proof of Theorem~\ref{thm:main}. To this end we need a result in~\cite{abs}.
Let $C$ be an Eulerian cycle in a 4-regular labeled graph $H$ and let $\mathbf{w}$ be the induced double occurrence word.\ Recall that two
vertices $u$ and $v$, 
labeled  
$a$ and $b$, respectively, say, are interlaced in $\mathbf{w}$ if
$\mathbf{w}=\mathbf{u} 
a\mathbf{x} 
b\mathbf{y} a \mathbf{z}$ for some (possibly empty) intervals 
$\mathbf{u}$, $\mathbf{x}$, 
$\mathbf{y}$ and $\mathbf{z}$ of $\mathbf{w}$.\ The $uv$-transposition of $\mathbf{w}$ is the word $\mathbf{w}^{uv}=\mathbf{u} 
a\mathbf{y} b\mathbf{x} a \mathbf{z}$ \cite{abs}. Thus an $uv$-transposition of $\mathbf{w}$
amounts to replace one of the the subpaths of $C$ connecting $u$ and $v$ with
he other one. The relation between $uv$-transposition and $uv$ pivoting is given in the next lemma which specializes a more general result in \cite{abs} (see also \cite{fom}).
\begin{lemma}\label{th:abu} Let $H$ be a 4-regular graph and let $\mathbf{w}$ be any of the double occurrence words it induces. Further, let $G(H,\mathbf{w})$ denote the interlacement graph of $\mathbf{w}$. Suppose that $G(H,\mathbf{w})$ is a bipartite graph. Then, for any edge $uv$ of $G(H,\mathbf{w})$ of $H$, one has  $G(H,\mathbf{w})^{uv}=G(H,\mathbf{w}^{uv})$.
\end{lemma}

\mybreak 
\textbf{Second proof of Theorem~\ref{thm:main}.} If $G$ is the fundamental graph of a series--parallel
graph, then $M^G$ is a binary matroid with no $M(K_4)$
minor by Dirac and Duffin's characterization. Dominoes are fundamental graphs of $K_4$ and holes can
be pivoted to either dominoes or $C_6$ (recall Lemma~\ref{lem:1})---notice that $C_6$
is a fundamental graph of $K_4$ as well (Figure~\ref{fi:fundgraph})--it follows that $G$ is
BDH-free by Lemma~\ref{lem:matorbit}. Conversely, if $G$ is BDH, then by Theorem~\ref{thm:mosabis} (in the language of Lemma~\ref{th:abu}), $G\cong G(m(H),\mathbf{w})$ for some series--parallel graph $H$ (observe that $m(H)$ is a 4-regular graph) and some code $\mathbf{w}$. By
Lemma~\ref{th:abu}, pivoting on edges $G$
affects neither $H$ nor $m(H)$. Consequently, every graph
in the orbit of $[G]$ is
a BDH. Therefore $M^G$
has no $M(K_4)$ minor by Lemma~\ref{lem:matorbit} and Lemma~\ref{lem:matroidorbit} and $G$ is a fundamental graph of such a
matroid and therefore the fundamental graph of a series--parallel graph.\hspace*{\fill}$\Box$

\end{document}